%% file: DisplayAdvertising_ARXIV.tex
\RequirePackage{etoolbox}
\newtoggle{shortversion}

\newcommand{\inshort}[1]{\iftoggle{shortversion}{#1}{}} 
\newcommand{\inlong}[1]{\iftoggle{shortversion}{}{#1}}  

\documentclass[11pt]{article}

\usepackage{amsmath,amsfonts,amssymb,tikz}
\usepackage[margin=1in]{geometry}
\usepackage{amsthm} 
\usepackage{csquotes} 
\usepackage{nicefrac}
\usepackage{paralist}
\usepackage[ruled,vlined,linesnumbered]{algorithm2e}
\usepackage{float}
\usepackage[font=small,margin=15pt]{caption}
\usepackage{forloop}

\newtheorem{theorem}{Theorem}[section]
\newtheorem{observation}[theorem]{Observation}

\newtheorem{lemma}[theorem]{Lemma}
\newtheorem{corollary}[theorem]{Corollary}

\newtheorem{claim}[theorem]{Claim}
\newtheorem{proposition}[theorem]{Proposition}
\newtheorem{definition}[theorem]{Definition}

\makeatletter
\newtheorem*{rep@theorem}{\rep@title}
\newcommand{\newreptheorem}[2]{%
\newenvironment{rep#1}[1]{%
 \def\rep@title{#2 \ref{##1}}%
 \begin{rep@theorem}}%
 {\end{rep@theorem}}}
\makeatother

\newreptheorem{theorem}{Theorem}

\newcommand{\defcal}[1]{\expandafter\newcommand\csname c#1\endcsname{{\mathcal{#1}}}}
\newcounter{calBbCounter}
\forLoop{1}{26}{calBbCounter}{
    \edef\letter{\Alph{calBbCounter}}
    \expandafter\defcal\letter
}

\newcommand{\R}{{\mathbb{R}}}
\newcommand{\eps}{\varepsilon}
\newcommand{\smax}{\max \nolimits^{(2)}}
\newcommand{\G}{\mathcal{G}}
\newcommand{\DSP}{\ensuremath{\mathbf{DSP}}}
\newcommand{\cappart}{\times}

\begin{document}

\title{Distributed Signaling Games\inshort{\\[-2mm]{\large (Extended Abstract)}}}

\author{Moran Feldman\thanks{Microsoft Research, Israel, and EPFL, {\tt moran.feldman@epfl.ch}.}
\and Moshe Tennenholtz\thanks{Microsoft Research, Israel, and Technion-IIT, {\tt moshet@ie.technion.ac.il}. }
\and Omri Weinstein\thanks{Microsoft Research, Israel, and Princeton University, {\tt oweinste@cs.princeton.edu}.
Research supported by a Simons fellowship in Theoretical Computer Science.}}

\maketitle

\input{abstract.tex}

\thispagestyle{empty} \setcounter{page}{0}
\newpage

\input{newintroJune24.tex}


\input{prelims}


\input{computational}


\inlong{\input{shapley}}


\inlong{\input{PoA}}


\input{discussion}


\bibliographystyle{alpha}
\bibliography{refs}

\end{document}

%% file: abstract.tex
\begin{abstract}

A recurring theme in recent computer science literature is that proper design of signaling schemes is a crucial aspect 
of effective mechanisms aiming to optimize social welfare or revenue. One of the research endeavors of this line of work is understanding the algorithmic and computational complexity of
designing efficient signaling schemes. In reality, however, information is typically 
not held by a central authority, but is distributed among multiple sources (third-party ``mediators"), a fact that dramatically changes the strategic and combinatorial nature of the signaling problem, making it a game between information providers, as opposed to a traditional mechanism design problem.

In this paper we introduce {\em distributed signaling games},  while using display advertising as a canonical example for introducing this 
foundational framework. A distributed signaling game may be a pure coordination game (i.e., a distributed optimization task), or a non-cooperative game. 
In the context of pure coordination games, 
we show a wide gap between the computational complexity of the centralized and distributed signaling problems, proving that 
distributed coordination on revenue-optimal signaling is a much harder problem than 
its ``centralized" counterpart.
On the other hand, we show that if the information structure of each mediator is assumed to be ``local", then 
there is an efficient algorithm that finds a near-optimal ($5$-approximation) distributed signaling scheme.

In the context of non-cooperative games, the outcome generated by the mediators' signals may have different value to each. 
The reason for that is typically the desire of the auctioneer to align the incentives of the mediators with his own by a compensation relative to the marginal benefit from their signals. We design a mechanism for this problem via a novel application of Shapley's value, and show that it possesses some interesting properties; in particular, it always admits a pure Nash equilibrium, and it never decreases the revenue of the auctioneer.

\end{abstract}

%% file: newintroJune24.tex
\section{Introduction}

The topic of signaling has recently received much attention in the computer science literature on mechanism design
\cite{CCDT14, DIR14, Dughmi14, Emek14}. A recurring theme of this literature is that proper design of a signaling scheme is crucial for 
obtaining efficient outcomes,
such as social welfare maximization or revenue maximization. 
In reality, however, sources of information
are typically not held by a central authority, but are rather distributed among third party mediators/information providers, a fact which dramatically changes the setup to be studied, making it a game between information providers rather than a more classic mechanism design problem. Such a game is in the spirit of work on the theory of teams in economics \cite{MarschakRadner},  whose computational complexity remained largely unexplored. 
The goal of this paper is to initiate an algorithmic study of such games, which we term {\em distributed signaling games}, 
via what we view as a canonical example: Bayesian auctions; and more specifically, display advertising in the presence of third party external 
mediators (information provides), capturing the (multi-billion) ad exchanges industry.   

Consider a web-site owner that auctions each user's visit to its site, a.k.a. impression. The impression types are assumed to arrive from a commonly known distribution. The bidders are advertisers who know that distribution, but only the web site owner knows the impression type instantiation, consisting of identifiers such as age, origin, gender and salary of the web-site visitor. As is the practice in existing ad exchanges, we assume the auction is a second price auction. The web-site owner decides on the information (i.e., signal) about the instantiation to be provided to the bidders, which then bid their expected valuations for the impression given the information provided. The selection of the proper signaling by the web-site is a central mechanism design problem. Assume, for example, an impression associated with two attributes: whether the user is male or female on one side, and whether he is located in the US or out of the US on the other side. This gives 4 types of possible users. Assume for simplicity that the probability of arrival of each user type is 1/4, and that there are four advertisers each one of them has value of \$100 for a distinguished user type and \$0 for the other types, where these values are common-knowledge. One can verify that an auctioneer who reveals no information receives an expected payoff of \$25, an auctioneer who reveals all information gets no payoff, while partitioning the impression types into two pairs, revealing only the pair of the impression which was materialized (rather than the exact instantiation) will yield a payoff of \$50, which is much higher revenue.

While the above example illustrates the importance of signaling, and its natural fit to mechanism design, its major drawback is 
in the unrealistic manner in which information is manipulated:
while some information about the auctioned item is typically published by the ad network \cite{YWZ13} (such information 
is modeled here as a public prior), 
and despite the advertisers' effort to perform ``behavioral targeting" by clever data analysis (e.g., utilizing  the browsing history of a specific user to 
infer her interests), the quantity of available contextual information and market expertise is often way beyond the capabilities of both advertisers and 
auctioneers.  This reality gave rise to ``third-party'' companies which develop technologies for 
collecting data and online statistics used to infer the contexts of auctioned impressions (see, e.g., \cite{MayerMitchell} and references therein).
Consequently, a new \emph{distributed} ecosystem has emerged, in which many third-party companies operate within the market
aiming at maximizing their own utility (royalties or other compensations), while significantly increasing the effectiveness of display advertising, as 
suggested by the following article recently published by 
Facebook:

\blockquote{\emph{``Many businesses today work with third parties such as Acxiom, Datalogix, and Epsilon to help manage and understand their marketing efforts. 
For example, an auto dealer may want to customize an offer to people who are likely to be in the market for a new car. The dealer also might want to 
send offers, like discounts for service, to customers that have purchased a car from them. To do this, the auto dealer works with a third-party company to 
identify and reach those customers with the right offer".}\\
\;\;\; \small{($\mathsf{www.facebook.com}$, ``Advertising and our Third-Party Partners", April 10, 2013.)}}

Hence, in reality sources of information are distributed. Typically, the information is distributed among several mediators or information providers/brokers, and is not held (or mostly not held) by a central authority/web-site owner.  In the display advertising example one information source may know the gender and one may know the location of the web-site visitor, while the web-site itself often lacks the capability to track such information.  The information sources need to decide on the communicated information. In this case the information sources become players in a game. To make the situation clearer, assume (as above) that the value of each impression type for each bidder is public-knowledge (as is typically the case in repeated interactions through ad exchanges which share their logs with the participants), and the only unknown entity is the instantiation of the impression type;  given the information learned from the information sources each bidder will bid his true expected valuation; hence, the results of this game are determined solely by the information providers. Notice that if, in the aforementioned example, the information provider who knows the gender reveals it while the other reveals nothing, then the auctioneer recieves a revenue of \$50 as in the centralized case, while the cases in which both information providers reveal their information or none of them do so result in lower revenues. This shows the subtlety of the situation. 

The above suggests that a major issue to tackle is the study of {\em distributed signaling games}, going beyond the realm of classical mechanism design. We use a model of the above display advertising setting, due to its centrality, as a tool to introduce this novel foundational topic.  The distributed signaling games  may be pure coordination games (a.k.a. distributed optimization), or non-cooperative games. In the context of pure coordination games each information source has the same utility from the output created by their joint signal. Namely, in the above example if the web-site owner pays each information source proportionally to the revenue obtained by the web-site owner then the aims of the information sources are identical. The main aim of the third parties/mediators is to choose their signals based on their privately observed information in a  distributed manner in order to optimize their own payoffs. Notice that in a typical embodiment, which we adapt, due to both technical and legal considerations, the auctioneer does not synthesize reported signals into new ones nor the information providers  are allowed to explicitly communicate among them about the signals, but can only broadcast information they individually gathered. The study of the computational complexity of this highly fundamental problem is the major technical challenge tackled in this paper. Interestingly, we show a wide gap between the computational complexity of the centralized and of the distributed signaling setups, proving that coordinating on optimal signaling is a much harder problem than the one discussed in the context of centralized mechanism design. On the other hand we also show a natural restriction on the way information is distributed among information providers, which allows for  an efficient constant approximation scheme.  

In the context of non-cooperative games the outcome generated by the information sources' reports may result in a different value for each of them. The reason for that is typically the desire of the auctioneer to align the incentives of the mediators with his own by a compensation relative to the marginal benefit from their signals. In the above example one may compare the revenue obtained without the additional information sources, to what is obtained through their help, and compensate relatively to the Shapley values of their contributions, which is a standard (and rigorously justified) tool to fully divide a gain yielded by the cooperation of several parties. Here we apply such division to distributed signaling games, and show that it possesses some interesting properties: in particular the corresponding game has a pure strategy equilibrium, a property of the Shapley value which is shown for the first time for signaling settings (and is vastly different from previous studies of Shapley mechanisms in non-cooperative settings such as cost-sharing games~\cite{Roughgarden}).

\subsection{Model} \label{sec_model}
Our model is a generalization of the one defined in \cite{GNS07}. 
There is a ground set $I = [n]$ of potential items (contexts) to be sold and a set $B = [k]$ of bidders. The value of 
item $j$ for bidder $i$ is given as $v_{ij}$. Following the above discussion (and the previous line of work, e.g., \cite{KEW07,GNS07}), 
we assume the valuation matrix $V=\{v_{i,j}\}$ is publicly known. 
An auctioneer is selling a single random item $j_R$, distributed according to some publicly known prior distribution $\mu$ over $I$, 
using a second price auction (a more detailed description of the auction follows).
There is an additional set  $M = [m]$ of ``third-party'' mediators. Following standard practice in game theoretic information models~\cite{Aumann76,Moses95,Emek14},
we assume each  mediator $t \in M$ is equipped with a {\it partition} (signal-set)  $\cP_t \in \Omega(I)$\footnote{For a set $S$, $\Omega(S) \triangleq \{\cA \subseteq 2^S \mid \bigcup_{A \in \cA} A = S, \forall_{A,B\in \cA} A\cap B = \varnothing\}$ is the collection of all partitions of $S$.}. Intuitively, $\cP_t$ captures the extra information $t$ has about the item which is about to be sold---he knows the set $S \in \cP_t$ to which the item $j_R$ belongs, but has no further 
knowledge about which item of $S$ it is (except for the a priori distribution $\mu$)---in other words, the distribution $t$ has in mind is $\mu|_S$.
For example, if the signal-set partition $\cP_t$ partitions the items of $I$ into pairs, then mediator $t$ knows to which pair $\{j_1, j_2\} \in \cP_t$ the item $j_R$ belongs, but 
he has no information whether it is $j_1$ or $j_2$, and therefore, from her point of view, $\Pr[j_R = j_1] = \mu(j_1) / \mu(\{j_1, j_2\})$. 

Mediators can signal some (or all) of the information they own to the network.
Formally, this is represented by allowing each mediator $t$ to report {\it any} super-partition $\cP_t'$, which is obtained
by  merging partitions in her signal-set partition $\cP_t$ (in other words $\cP_t$ must be a refinement of $\cP'_t$). In other words, a mediator may report any partition $\cP_t' $
for which there exists a set $\cQ'_t \in \Omega(\cP_t)$ such that $\cP'_t = \{\cup_{S \in A} S ~|~ A \in \cQ'_t\}$.
In particular, a mediator can always report $\{I\}$, in which case we say that he remains silent since he does not contribute any information.
The signals $\cP_1' ,\cP_2' ,\ldots , \cP_m '$ reported by the mediators  are \emph{broadcasted}\footnote{By saying that a mediator reports $\cP'_t$, we mean that he reports the bundle $S\in \cP'_t$ for which $j_R \in S$. The reader may wonder why our model is a broadcast model, and does not allow the mediators to report their information to the auctioneer through private channels, in which case the ad network will be able to manipulate and publish whichever information that best serves its interest. The primary reason for the broadcast assumption is that the online advertising market is highly dynamic and mediators often ``come and go'', so implementing such ``private contracts'' is infeasible. The second reason is that real-time bidding environments cannot afford the latency incurred by such a two-phase procedure in which the auctioneer first collects the information, and then selectively publishes it.
The auction process is usually treated as a ``black box", and modifying it harms the modularity of the system.} to the bidders,  inducing a combined partition $\cP \triangleq \times_{t=1}^m \cP'_t = \{ \cap_{i \in M} A_i \mid A_i \in \cP'_i \}$, which we call the {\em joint partition} (or {\em joint signal}).
$\cP$ splits the auction into separate ``restricted'' auctions. For each bundle $S\in \cP$, the item $j_R$ belongs to $S$ with probability $\mu(S)=\sum_{j\in S} \mu(j)$, in which case $S$ is signaled to 
the  bidders and a second-price auction is performed over $\mu|_S$. 
Notice that if the signaled bundle is 
$S\subseteq I$, then the (expected) value of bidder $i$ for $j_R \sim \mu|_S$
is $v_{i,S} = \frac{1}{\mu(S)}\sum_{j\in S} (\mu(j) \cdot v_{ij})$, and the truthfulness of the second price auction implies that this will also be bidder $i$'s bid for the restricted auction. The winner of the auction is the bidder with the maximum bid $\max_{i \in B} v_{i, S}$, and he is charged the second highest valuation for that bundle $\smax_{i \in B} v_{i, S}$. Therefore, the auctioneer's revenue with respect to $\cP$ is the expectation (over $S\in_R \cP$) of the price paid by 
the winning bidder:
\[ R(\cP) = \sum_{S\in\cP} \mu(S)\cdot \smax_{i \in B}(v_{i, S}) \enspace. \]

The joint partition $\cP$ signaled by the mediators can dramatically affect the revenue of the auctioneer.  
Consider, for example, the case where $V$ is the $4\times 4$ identity matrix, $\mu$ is the uniform distribution, and 
$M$ consists of two mediators associated with the partitions $\cP_1 = \{\{1,2\} \;,\;\{ 3,4 \}\}$ and $\cP_2 = \{\{1,3\} \;,\;\{ 2,4 \}\}$. 
If both mediators remain silent, the revenue is $R(\{I\}) = 1/4$ (as this is the average value of all $4$ bidders for a random item).
However, observe that $\cP_1\times\cP_2 = \{\{1\},\{2\},\{3\},\{4\}\}$, and the second highest value in every column of $V$ is $0$,
thus, if both report their partitions, the revenue drops to $R(\cP_1\times\cP_2) = 0$. Finally, if mediator
$1$ reports $\cP_1$, while meditor $2$ keeps silent, the revenue increases from $1/4$ to $R(\cP_1) = 1/2$, as the value of each pair of items is $1/2$ for two different bidders (thus, the second highest price for each pair is $1/2$). This example can be easily generalized to show that in general the intervention of mediators can
increase the revenue by a factor of $n/2$ !

Indeed, the purpose of this paper is to understand how mediators' (distributed) signals affect the revenue of the auctioneer.
We explore the following two aspects of this question: 
\begin{enumerate}
\item (Computational) Suppose the auctioneer has control over the signals reported by the mediators. We study the computational complexity 
of the following problem. {\it Given a $k\times n$ matrix $V$ of valuations and mediators' partitions $\cP_1 ,\cP_2 ,\ldots , \cP_m$, what is the revenue 
maximizing joint partition $\cP = \cP'_1\times\ldots \times \cP'_m$?} 
We call this problem the \emph{Distributed Signaling Problem}, and denote it by $\mathbf{DSP}(n,k,m)$.

We note that the problem studied in \cite{GNS07} is a special case of {\DSP}, in which there is a \emph{single} mediator ($m = 1$) with \emph{perfect knowledge} 
about the item sold and can report any desirable signal (partition).\footnote{In other words, $\cP_1$ is the partition of $I$ into singletons.}  

\item (Strategic) What if the auctioneer cannot control the signals reported by the mediators (as the reality of the problem usually entails)?
{\it Can the auctioneer introduce compensations 
that will incentivize mediators to report signals leading to increased revenue in the auction, when each mediator is acting selfishly?}

This is a mechanism design problem: 
Here the auctioneer's goal is to design a payment rule (i.e., a mechanism) 
for allocating (part of) his profit from the auction among the mediators, based on their reported signals and the auction's outcome, 
so that global efficiency (i.e., maximum revenue) emerges from their signals.
\end{enumerate}
Section~\ref{sc:our_results} summarizes our findings regarding the two above problems. 

\subsection{Our Results} \label{sc:our_results}

Ghosh et al.~\cite{GNS07} showed that computing the revenue-maximizing signal in their ``perfect-knowledge'' setup
is $NP$-hard, but present an efficient algorithm for computing a $2$-approximation of the optimal signal (partition). 
We show that when information is distributed, the problem becomes much harder.
More specifically, we present a gap-preserving reduction from the \emph{Maximum Independent Set} problem to {\DSP}.
\begin{theorem}[Hardness of approximating {\DSP}]\label{cor:hardness_DSP_IS}
If there exists an $O(m^{1/2-\eps})$ approximation (for some constant $\eps > 0$) for instances of $\DSP(2m, m + 1, m)$, 
then there exists a $O(N^{1-2\eps})$ approximation for Maximum Independent Set (MIS$_N$), where $N$ is the number of nodes in 
the underlying graph of the MIS instance.
\end{theorem}
\noindent Since the Maximum Independent Set problem is NP-hard to approximate to within a factor of $n^{1-\rho}$ for any fixed $\rho>0$ \cite{Hastad01},  
Theorem~\ref{cor:hardness_DSP_IS} indicates that approximating the revenue-maximizing 
signal, even within a multiplicative factor of $O((\min\{n, k, m\})^{1/2-\eps})$, is NP-hard. 
In other words, one cannot expect a reasonable approximation ratio for $\DSP(n, k, m)$ when the three parameters of the problem are all ``large''. The next theorem shows that a ``small'' value for either one of the parameters $n$ or $k$ indeed implies a better approximation ratio.
\begin{theorem}[Approximation algorithm for small $n$ or $k$] \label{thm:simple_approximations}
There is an efficient $\max\{1, \min\{n, k - 1\}\}$-approximation algorithm for {\DSP}.
\end{theorem}

We leave open the problem of determining whether one can get an improved approximation ratio when the parameter $m$ is ``small''. For $m = 1$, the result of~\cite{GNS07} implies immediately a $2$-approximation algorithm. However, even for the case of $m = 2$ we are unable to find an algorithm having a non-trivial approximation ratio. We mitigate
the above results by proving that for a natural (and realistic) class of mediators called \emph{local experts}\inlong{ (defined in Section~\ref{sec_hardness_DkS})}, 
there exists an efficient $5$-approximation algorithm for \DSP{}.

\begin{theorem}[A $5$-approximation algorithm for Local Expert mediators] \label{thm:alg_local_experts}
If mediators are local experts, 
there exists an efficient $5$-approx\-imation algorithm for {\DSP}.
\end{theorem}

In the strategic setup, we design a fair (symmetric) payment rule $\cS: (\cP'_1, \cP'_2,...,\allowbreak \cP'_m) \rightarrow \mathbb{R}_+^m$ 
for incentivizing mediators to report useful information they own, and refrain from reporting information with negative impact on the revenue. 
This mechanism is inspired by the  
\emph{Shapley Value}---it distributes part of the auctioneer's surplus among the mediators according to their expected relative marginal 
contribution to the revenue, when ordered randomly.\footnote{Shapley's value was originally introduced in the context of cooperative games,
where there is a well defined notion of a coalition's value. In order to apply this notation to a non-cooperative game, we assume the game has
some underlying global function ($v(\cdot)$) assigning a value to every strategy profile of the players, and the Shapley value of each player is defined with 
respect to $v(\cdot)$. In this setting, a ``central planner" (the auctioneer in our case) is the one making the utility transfer to the ``coalised" players.
For the formal axiomatic definition of a value function and Shapley's value function, see \cite{Shapley53}.} We first show that this mechanism 
always admits a \emph{pure} Nash equilibrium, a property we discovered to hold for arbitrary games where the value of the game is distributed 
among players according to Shapley's value function.

\begin{theorem} \label{thm_pure_eq}
Let $\G_m$ be a non-cooperative $m$-player game in which 
the payoff of each player is set according to $\cS$. Then
$\G_m$ admits a pure Nash equilibrium. Moreover, best response dynamics are guaranteed to converge to such an equilibrium.
\end{theorem}


We then turn to analyze the revenue guarantees of our mechanism $\cS$. Our first theorem shows that using the mechanism $\cS$
never decreases the revenue of the auctioneer compared to the initial state (i.e., when all mediators are silent).
\begin{theorem} \label{th:better_then_silence}
For every Nash equilibrium $(\cP'_1, \cP'_2, \ldots, \cP'_m)$ of $\cS$, $R(\cappart_{t \in M} \cP'_t) \geq R(\{I\})$.
\end{theorem}

\noindent The next two theorems provide tight bounds on the \emph{price of anarchy} and \emph{price of stability} of $\cS$.%
\footnote{The price of anarchy (stability) is the ratio between the revenue of the optimum and the worst (best) Nash equilibrium.}
Unlike in the computational setup, restricting the mediators to be local experts does enable us to get improved results here.

\begin{theorem} \label{thm_poa_ub_combined} 
The price of anarchy of $\cS$ under any instance $\mathbf{DSP}(n,k,m)$ is no more than $\max\{1, \min \{k-1, n\}\}$.
\end{theorem}

\begin{theorem} \label{thm_unbounded_pos} 
For every $n \geq 1$, there is a $\DSP(3n + 1, n + 2, 2)$ instance for which the price of stability of $\cS$ is at least $n$. Moreover, all the mediators in this instance are local experts.
\end{theorem}

Interestingly, an adaptation of Shapley's uniqueness theorem \cite{Shapley53} to our non-cooperative setting asserts that the price of anarchy of our mechanism is inevitable if one insists on a few natural requirements---essentially anonymity and efficiency\footnote{I.e., the sum of payments is equal to the total surplus of the auctioneer.} 
of the payment rule---and assuming the auctioneer alone can introduce payments. 
\inlong{We discuss this further in Section~\ref{ssc:inevitable}.}


\subsection{Additional Related Work}

The formal study of internet auctions with contexts was introduced by
\cite{KEW07} where the authors studied the impact of contexts in the related Sponsored Search model, 
and showed that \emph{bundling contexts} may have a significant impact on the revenue
of the auctioneer. 
The subsequent work of Ghosh et. al.~\cite{GNS07} considered the computational algorithmic problem of computing the revenue maximizing partition 
of items into bundles, under a second price auction in the full information setting. Recently, Emek et al.~\cite{Emek14} studied signaling (which generalizes bundling) in the context of 
display advertising. They explore the computational complexity of computing a signaling scheme that maximizes the auctioneer's 
revenue in a Bayesian setting. Unlike our distributed setup, both models of~\cite{GNS07} and~\cite{Emek14}
are \emph{centralized}, in the sense that the auctioneer has full control over the bundling process (which in our terms corresponds to 
having a single mediator with a perfect knowledge about the item sold).


A different model with knowledgeable third parties was recently considered by Cavallo et al.~\cite{CMV15}. However, the focus of this model is completely different then ours. More specifically, third parties in this model use their information to estimate the clicks-per-impression ratio, and then use this estimate to bridge between advertisers who would like to pay-by-click and ad networks which use a pay-by-impression payment scheme.

%% file: prelims.tex
\section{Preliminaries}\label{sec_prelims}

Throughout the paper we use capital letters for sets and calligraphic letters for set families. For example, the partition $\cP_t$ representing the knowledge 
of mediator $t$ is a set of sets, and therefore, should indeed be calligraphic according to this notation.
\inlong{%
A mechanism $\cM$ is a tuple of payment functions $(\Pi_1, \Pi_2, \ldots, \Pi_m)$ determining the compensation of every mediator given a strategy profile (i.e., $\Pi_t : \Omega(\cP_1)\times\Omega(\cP_2)\times\ldots\times\Omega(\cP_m) \longrightarrow \R^+$). Every mechanism $\cM$ induces the following game between mediators.

\begin{definition}[\DSP{} game] \label{def_gm}
Given a mechanism $\cM = (\Pi_1, \Pi_2, \ldots, \Pi_m)$ and a $\DSP(n, k, m)$ instance, the $\DSP_{\cM}(n, k, m)$ game is defined as follows.
Every mediator $t \in M$ is a player whose strategy space consists of all partitions $\cP_t'$ for which $\cP_t$ is a refinement. Given a strategy profile $\cP_1' ,\cP_2' ,\ldots , \cP_m'$, the payoff of mediator $t$ is
$\Pi_t (\cP_1' ,\cP_2' ,\ldots , \cP_m')$.
\end{definition}

%
} 
Given a $\DSP$ instance and a set $S \subseteq I$, we use the shorthand $v(S) := \smax_{i \in B}(v_{i, S})$
to denote the second highest bid in the restricted auction $\mu|_S$. Using this notation, the expected revenue of the auctioneer
under the (joint) partition $\cP$ of the mediators can be written as
\[
	R(\cP)
	=
	\sum_{S \in \cP} \mu(S) \cdot v(S)
	\enspace.
\]

\inlong{
For a $\DSP_\cM$ game, let $\cE(\cM)$ denote the set of Nash equilibria of this game and let $\cP^*$ be a maximum revenue strategy profile. 
The \emph{Price of Anarchy} and \emph{Price of Stability}  of $\DSP_\cM$ are defined as:
\begin{align*}
	&
	PoA := \max_{\cP \in \cE(\cM)} \frac{R(\cP^*)}{R(\cP)}
	\enspace,
	\qquad \text{and} \qquad
	PoS := \min_{\cP \in \cE(\cM)} \frac{R(\cP^*)}{R(\cP)}
	\enspace,
\end{align*}
respectively. Notice that our definition of the price of anarchy and price of stability differs from the standard one by using revenue instead of social welfare.
} 

\inshort{
\paragraph{Paper Organization.}
Due to space constraints we only give the proof for Theorem~\ref{cor:hardness_DSP_IS} (in Section~\ref{sec:hardness_DSP_IS}). Section~\ref{sec:discussion} summarizes our contributions and discuss possible avenues for future research.
}

%% file: computational.tex
\inlong{
\section{The Computational Problem} \label{sec_hardness_DkS}

This section explores {\DSP} from a pure combinatorial optimization viewpoint. In other words, 
we assume the auctioneer can control the signals produced by each mediator. The objective of the auctioneer 
is then to choose a distributed strategy profile $\cP_1' ,\cP_2' ,\ldots , \cP_m '$ whose combination $\cappart_{t} \cP_t'$
yields maximum revenue in the resulting auction. We begin with negative results (proving Theorem \ref{cor:hardness_DSP_IS} in 
Subsection~\ref{sec_hardness_dsp}), and then proceed with a few approximation algorithms for the problem (Subsection~\ref{sec_gen_algorithms}), including a
$5$-approximation algorithm for the case of \emph{local expert} mediators (Theorem \ref{thm:alg_local_experts}).

\subsection{Hardness of approximating \DSP} \label{sec_hardness_dsp}
} 

\inshort{
\section{Proof of Theorem~\ref{cor:hardness_DSP_IS}} \label{sec:hardness_DSP_IS}
}

It is not hard to show that solving {\DSP} exactly, i.e., finding the maximum revenue strategy profile, is
NP-hard. In fact, this statement directly follows from the NP-hardness result of \cite{GNS07}
by observing that the special case of a single mediator with perfect knowledge 
($\cP_1 = \{\{j\} \mid j \in I\}$) is equivalent to the centralized signaling model of \cite{GNS07}.

%

\begin{proposition} \label{thm_dob_nphard}
Solving {\DSP} exactly is NP-Hard. 
\end{proposition}

The main result of this section is that approximating the revenue-maximizing signals, even within a multiplicative factor of 
$O((\min\{n, k, m\})^{1/2-\eps})$, remains NP-hard. This result is achieved by a 
gap preserving reduction from Maximum Independent Set to \DSP.

\begin{theorem} \label{thm:general_model_reduction}
For every integer $\ell \geq 1$, an $\alpha \geq 1$ approximation for $\mathbf{DSP}(2\ell N, \allowbreak \ell N + 1, \ell N)$ 
induces a $\alpha(1 + \ell^{-1}(N + 1))$ approximation for the Maximum Independent Set problem (where $N$ is the number of nodes in the graph).
\end{theorem}

Observe that Theorem~\ref{cor:hardness_DSP_IS} follows easy from Theorem~\ref{thm:general_model_reduction}.

\begin{reptheorem}{cor:hardness_DSP_IS}
If there exists an $O(m^{1/2-\eps})$ approximation (for some constant $\eps > 0$) for instances of $\DSP(2m, m + 1, m)$, 
then there exists a $O(N^{1-2\eps})$ approximation for Maximum Independent Set (MIS$_N$), where $N$ is the number of nodes in 
the underlying graph of the MIS instance.
\end{reptheorem}
\begin{proof}
Let $\ell = N + 1$. By Theorem~\ref{thm:general_model_reduction}, there exists an approximation algorithm for Maximum Independent Set whose approximation ratio is $2\alpha$, where $\alpha$ is the approximation ratio that can be achieved for instances of $\DSP(2\ell N, \ell N + 1, \ell N)$. On the other hand, by our assumption:
\[
	\alpha
	=
	O((\ell N)^{1/2 - \eps})
	=
	O(N^{2(1/2 - \eps)})
	=
	O(N^{1 - 2\eps})
	\enspace,
\]
which completes the proof.
\end{proof}	

In the rest of this section we prove Theorem~\ref{thm:general_model_reduction}. The high-level idea of the reduction is as follows. Given a graph $G= (V,E)$, we map it to a {\DSP} instance by 
associating $\ell$ 
pairs of (equi-probbable) items $\{(j_{v,k}, j'_{v,k})\}_{k=1}^\ell$ with each node of the graph $v\in V$, where the items
$j'_{v,k}$ are auxiliary items called the ``helper items" of node $v$.  
Additionally, for each node we have $\ell$ single-minded bidders, each of which is interested (exclusively) in a specific item $j_{v,k}$. An
additional bidder ($i_h$) is interested in \emph{all} helper-items of all nodes. 
Finally, each node $v \in V$ has a corresponding set of $\ell$ mediators $\{m_{v, k}\}_{k=1}^\ell$ (one for each items pair). Each mediator $m_{v,k}$ has a single bit of information, corresponding to whether or not the auctioned item belongs to the union set of \emph{the $k$-th pair $(j_{v,k}, j'_{v,k})$ together with all helper elements of 
the neighbor nodes of $v$}. 

As the helper elements are valuable only to a \emph{single} bidder $i_h$, a part in a partition can have a non-zero contribution to the revenue only if it contains at least two non-helper elements or at least
one helper element and at least one non-helper element. However, our construction ensures that, whenever a strategy profile $\cP$ involves at least \emph{two ``active" mediators of neighboring nodes} $(u,v)\in E$, the resulting joint partition \emph{isolates the non-helper elements of $u$ from helper elements} (except for maybe one part that might contain elements of both types). Thus, in a high revenue strategy profile the set of nodes associated with many ``speaking" (active) mediators must be close to an independent set. 

We now proceed with the formal proof. Given a graph $G = (V, E)$ of $N$ nodes, consider an instance of {\DSP} consisting of the following:
\begin{itemize}
	\item A set $\{j_{v,k}, j'_{v,k} \mid v \in V, 1 \leq k \leq \ell\}$ of $2\ell N$ items having equal probabilities to appear. For every node $v$ and integer $1 \leq k \leq \ell$, the element $j'_{v,k}$ is referred to as the ``helper'' item of $j_{v,k}$. For notational convenience, we define the following sets for every node $v \in V$: $J_v = \{j_{v,k} \mid 1 \leq k \leq \ell\}$ and $H_v = \{j'_{v,k} \mid 1 \leq k \leq \ell\}$.
	\item A set $\{i_{v,k}, \mid v \in V, 1 \leq k \leq \ell\} \cup \{i_h\}$ of $\ell N + 1$ bidders. For every node $v$ and integer $1 \leq k \leq \ell$, the bidder $i_{v,k}$ has a value of $2\ell N$ for the item $j_{v,k}$ and a value of $0$ for all other item. The remaining bidder $i_h$ has a value of $2\ell N$ for each helper item (i.e., each item from the set $\{j'_{v,k} \mid v \in V, 1 \leq k \leq \ell\}$) and a value of $0$ for the non-helper item.
	\item There are $\ell N$ mediators $\{m_{v,k}, \mid v \in V, 1 \leq k \leq \ell\}$. For every node $v$ and integer $1 \leq k \leq \ell$ the partition of mediator $m_{v,k}$ is defined as follows:
	\[
		\cP_{v,k} = \left\{\{j_{v,k}, j'_{v,k}\} \cup \bigcup_{u \mid uv \in E} H_u\right\} \cup \{R_v\}
		\enspace,
	\]
	where $R_v$ is the set of remaining elements that do not belong to the first part of the partition. Informally, $\cP_{v,k}$ is a binary partition where on one side we have the two elements $j_{v,k}$ and $j'_{v,k}$ and all the helper elements corresponding to neighbor nodes of $v$, and on the other side we have the rest of the elements.
\end{itemize}

\inlong{Let us}\inshort{We} begin the analysis of the above {\DSP} instance by determining the contribution of each part in an arbitrary partition $\cP$ to $R(\cP)$.\inshort{ Due to space constraints, the (technical) proof of the next claim has been omitted from this extended abstract.}

\begin{claim} \label{cl_rev_contribution}
 Let $\cP$ be an arbitrary partition. Then, for every part $S \in \cP$, 
 the contribution of $S$ to $R(\cP)$ is $0$, unless $|S|\geq 2$ and $S$ contains at least one \emph{non-helper} item. 
 In the last case, the contribution of $S$ to the total revenue is:
 \[\mu(S) \cdot v(S) = 1 \enspace.\]
\end{claim}
\inlong{
\begin{proof}
There are a few cases to consider.
\begin{itemize}
	\item If $S$ contains only helper elements, then it is valuable only to $i_h$, and thus, has a $0$ contribution to $R(\cP)$.
	\item If $S$ contains only one element, then it is valuable only to one bidder because each element is valuable only to one element. Thus, it again contributes $0$ to $R(\cP)$.
	\item If $S$ contains multiple elements, at least one of which is \emph{non-helper}, then it is valuable to at least two bidders. Specifically, for every element $j_{v,k} \in S$, $v_{i_{v,k}, S} = 2\ell N / |S|$. Additionally, if there are helper elements in $S$, then:
	\[
		v_{i_h, S}
		=
		\frac{2\ell N \cdot |S \cap \{j'_{v,k} \mid v \in V, 1 \leq k \leq \ell\}|}{|S|}
		\geq
		\frac{2\ell N}{|S|}
		\enspace.
	\]
	Hence, for any such part $S$ we get $v(S) = \smax_{i \in B} v_{i, S} = 2\ell N / |S|$, and the contribution of the part to $R(\cP)$ is $\mu(S) \cdot v(S) = 1$. \qedhere
\end{itemize}
\end{proof}
} 

Let $A$ be an arbitrary independent set of $G$, and let $S_A := \{m_{v, k} \mid v \in A, 1 \leq k \leq \ell\}$. The following claim lower bounds the revenue of the joint partition arising when the mediators of $S_A$ are the only speaking mediators.

\begin{claim} \label{cl_rev_S_i}
$R(\cappart_{m_{v,k} \in S_A} \cP_{v,k}) \geq |S_A| = \ell \cdot |A|.$ 
\end{claim}
\begin{proof}
Observe that $\cP_{v,k}$ separates $j_{v,k}$ from every other item of the set $\{j_{v,k} \mid v \in V, 1 \leq k \leq \ell\}$. Hence, $\cappart_{m_{v,k} \in S_A} \cP_{v,k}$ contains $|S_A|$ different parts $\{T_{v,k} \mid m_{v,k} \in S_A\}$, where each part $T_{v,k}$ contains $j_{v,k}$. On the other hand, each pair $(j_{v,k}, j'_{v,k})$ of items is separated only by the partitions of mediators corresponding to neighbors of $v$. Since $A$ is independent, this implies that $j_{v,k}$ and $j'_{v,k}$ share part in $\cappart_{m_{v,k} \in S_A} \cP_{v,k}$ for every mediator $m_{v,k} \in S_A$. In other words, for every $m_{v,k} \in S_A$, the part $T_{v,k}$ contains $j'_{v,k}$ in addition to $j_{v,k}$, and thus by Claim \ref{cl_rev_contribution}, contributes $1$ to $R(\cappart_{m_{v,k} \in S_A} \cP_{v,k})$. Therefore, $R(\cappart_{m_{v,k} \in S_A} \cP_{v,k}) \geq |S_A| = \ell \cdot |A|$, as claimed. 
\end{proof}

Claim \ref{cl_rev_S_i} asserts that there exists a solution for the above {\DSP} instance whose value is at least $\ell \cdot OPT$, where $OPT$ is the size of the maximum independent set in $G$.

Consider now an arbitrary set $S$ of mediators, and let $S' = \{m_{v,k} \in S \mid \forall_{m_{u,k'} \in S}\, uv \not \in E\}$. Informally, a mediator $m_{v,k}$ is in $S'$ if it belongs to $S$ and no neighbor node $u$ of $v$ has a mediator in $S$. The following lemma upper bounds in terms of $|S'|$ the revenue of the joint partition of the mediators in $S$.
\begin{lemma} \label{lemma_rev_upper_bound} 
$R(\cappart_{m_{v,k} \in S} \cP_{v,k}) \leq |S'| + N + 1$.
\end{lemma}
\begin{proof}
Each partition $\cP_{v,k}$ separates a single non-helper element $j_{v,k}$ from the other non-helper elements. Hence, $\cappart_{m_{v,k} \in S} \cP_{v,k}$ consists of at most $|S| + 1$ parts containing non-helper elements (recall that parts with only helper elements have $0$ contribution to $R(\cappart_{m_{v,k} \in S} \cP_{v,k})$, and thus, can be ignored). Let us label these parts $\{T_{v,k}\}_{m_{v,k} \in S}, T$, where $T_{v,k}$ is the part containing $j_{v,k}$ for every $m_{v,k} \in S$ and $T$ is the part containing the remaining non-helper elements. Let us upper bound the contribution of each such part to $R(\cappart_{u \in S} \cP_u)$.
\begin{itemize}
	\item The part $T$ and all the parts $\{T_{v,k}\}_{m_{v,k} \in S'}$ can contribute at most $1$ each because no part has a larger contribution.
	\item Consider a part $T_{v,k}$ obeying $m_{v,k} \in S \setminus S'$ and $|J_v \cap S| = 1$ (i.e., $m_{v, k}$ is the only mediator of $v$ belonging to $S$). This part also contribute at most $1$, but there can be at most $N$ such parts, one for every node.
	\item Finally, consider a part $T_{v,k}$ obeying $m_{v,k} \in S \setminus S'$ and $|J_v \cap S| \geq 2$ (i.e., at least two mediators of $v$ belong to $S$). By the construction of  $\cP_{v,k}$, $T_{v,k}$ can contain in addition to $j_{v,k}$ only the corresponding helper element $j'_{v,k}$ and helper elements from $\bigcup_{u \mid uv \in E} H_u$. Let us see why none of these helper elements actually belongs to $T_{v,k}$, and thus, the part $T_v$ contains only $j_{v,k}$ and contributes $0$.
	\begin{itemize}
		\item Since $m_{v,k} \not \in S'$, there exists a neighbor node $v'$ of $v$ having a mediator $m_{v',k'} \in S$. Then, the partition $\cP_{v',k'}$ separates $j'_{v,k}$ from $j_{v,k}$ and guarantees that $j'_{v,k} \not \in T_{v,k}$.
		\item Let $m_{v,k'}$ be another mediator in $J_v \cap S$ (exists since $|J_v \cap S| \geq 2$). Then, the partition $\cP_{v,k'}$ separates the helper elements of $\bigcup_{u \mid uv \in E} H_u$ from $j_{v,k}$ and guarantees that none of these helper mediators belongs to $T_{v,k}$.
	\end{itemize}
\end{itemize}
In conclusion: $R(\cappart_{m_{v,k} \in S} \cP_{v,k}) \leq |S'| + N + 1$, as claimed.
\end{proof}

We are now ready to prove Theorem~\ref{thm:general_model_reduction}.
\begin{algorithm}
\caption{\textsf{Independent Set Algorithm}} \label{alg:independent_set}
	Construct a {\DSP} instance from the independent set instance as described above.\\
	Run the $\alpha$-approximation algorithm for {\DSP} assumed by Theorem~\ref{thm:general_model_reduction} on the constructed instance, and let $S$ be the set of mediators speaking in the obtained strategy profile.\\
	Calculate the configuration $S' = \{m_{v,k} \in S \mid \forall_{m_{u,k'} \in S}\, uv \not \in E\}$.\\
	Calculate the independent set $A' = \{v \in V \mid \exists_{1 \leq k \leq \ell} \, m_{v, k} \in S'\}$.\\
	If $A'$ is non-empty output it, otherwise output an arbitrary single node.
\end{algorithm}

\begin{proof}[Proof of Theorem~\ref{thm:general_model_reduction}]
Consider Algorithm~\ref{alg:independent_set}. We would like to show that this algorithm is an $\alpha(1 + \ell^{-1}(N + 1))$-approximation algorithm for Maximum Independent Set, which proves the theorem. The {\DSP} instance constructed by Algorithm~\ref{alg:independent_set} has a strategy profile of revenue at least $\ell \cdot OPT$ by Claim~\ref{cl_rev_S_i}. Since $S$ is obtained using an $\alpha$-approximation algorithm, $R(\cappart_{m_{v, k} \in S} \cP_{v, k}) \geq \ell \cdot OPT/\alpha$. By Lemma~\ref{lemma_rev_upper_bound} we now get: 
\[
	|S'|
	\geq
	R(\cappart_{m_{v, k} \in S} \cP_{v, k}) - (N + 1)
	\geq
	\ell \cdot OPT / \alpha - (N + 1)
	\enspace.
\]

Informally, $A'$ is the set of nodes having mediators in $S'$. The independence of $A'$ follows from the construction of $S'$, which guarantees that Algorithm~\ref{alg:independent_set} outputs an independent set. Additionally, since each node has $\ell$ mediators:
	\[
		|A'|
		\geq
		\frac{|S'|}{\ell}
		\geq
		\frac{\ell \cdot OPT / \alpha - (N + 1)}{\ell}
		=
		\frac{OPT}{\alpha} - \frac{N + 1}{\ell}
		\enspace.
	\]
If $OPT \geq \alpha(1 + \ell^{-1}(N + 1))$, then:
\begin{align*}
	|A'|
	\geq{} &
	\frac{OPT}{\alpha} - \frac{N + 1}{\ell}
	=
	\frac{OPT}{\alpha(1 + \ell^{-1}(N + 1))} + \frac{[1 - (1 + \ell^{-1}(N + 1))^{-1}] \cdot OPT}{\alpha} - \frac{N + 1}{\ell}\\
	\geq{} &
	\frac{OPT}{\alpha(1 + \ell^{-1}(N + 1))}
	\enspace.
\end{align*}
On the other hand, if $OPT \leq \alpha(1 + \ell^{-1}(N + 1))$, then the solution of Algorithm~\ref{alg:independent_set} is of size at least $\frac{OPT}{\alpha(1 + \ell^{-1}(N + 1))}$ simply because it is not empty.
\end{proof}


\inlong{
\subsection{Approximation algorithms for \DSP} \label{sec_gen_algorithms}

In light of Theorem \ref{thm:general_model_reduction}, 
an efficient algorithm with a reasonable approximation guarantee for general {\DSP} is unlikely to exist when the three parameters of the problem are all ``large''. Subsection~\ref{sec_alg_gen}
gives a trivial algorithm which has a good approximation guarantee when either $n$ or $k$ is small. A more 
interesting result is given in Subsection~\ref{sec_alg_local_experts}, which proves a $5$-approximation algorithm for
\DSP under the assumption that the mediators are \emph{local experts} (as stated in Theorem \ref{thm:alg_local_experts}).
 

\subsubsection{A simple $\max\{1, \min\{n,k-1\}\}$-approximation algorithm for \DSP} \label{sec_alg_gen}

In this section we prove the following theorem.

\begin{reptheorem}{thm:simple_approximations}
There is an efficient $\max\{1, \min\{n, k - 1\}\}$-approximation algorithm for {\DSP}.
\end{reptheorem}
\begin{proof}
We show that the algorithm that simply returns the partition $\{I\}$, the joint partition corresponding to the case where all mediators are silent, has the promised approximation guarantee. For that purpose we analyze the revenue of $\{I\}$ in two different ways:
\begin{itemize}
\item
Let $\cP' = (\cP'_1, \cP'_2, \ldots, \cP'_m)$ be an arbitrary strategy profile of the instance in question. The revenue of $\cP'$ is:
\begin{align*}
		R(\cappart_{t = 1}^m \cP'_t)
		={} &
    \sum_{S \in \cappart_{t = 1}^m \cP'_t} \mspace{-9mu} \mu(S) \cdot v(S)
    \leq
    |\cappart_{t = 1}^m \cP'_t| \cdot \max_{S \in \cappart_{t = 1}^m \cP'_t} \mu(S) \cdot v(S)\\
    \leq{} &
    n \cdot \max_{S \in \cappart_{t = 1}^m \cP'_t} \mu(S) \cdot v(S)
		\leq
		n \cdot R(\{I\})
    \enspace,
\end{align*}
where the last inequality holds since, for every set $S$, $R(\{I\}) = v(I) \geq v(S) \cdot \mu(S)$. 
This shows that the approximation ratio of the trivial strategy profile $\{I\}$ provides an $n$-approximation to the optimal 
revenue. 

\item

If $k = 1$, then the revenue of any strategy profile is $0$ since we assume a second price auction. Hence, we can assume from now on $k > 1$.%
Let $\cP' = (\cP'_1, \cP'_2, \ldots, \cP'_m)$ be an arbitrary strategy profile of the instance in question. The revenue of $\cP'$ is:\begin{align*}
		R(\cappart_{t = 1}^m \cP'_t)
		={} &
    \sum_{S \in \cappart_{t = 1}^m \cP'_t} \mspace{-9mu} \mu(S) \cdot v(S)
		=
		\sum_{S \in \cappart_{t = 1}^m \cP'_t} \mspace{-9mu} \mu(S) \cdot \left(\smax_{i \in B} \frac{\sum_{j \in S} \mu(j) \cdot v_{i,j}}{\mu(S)} 
			\right)\\
		={} &
		\sum_{S \in \cappart_{t = 1}^m \cP'_t} \left(\smax_{i \in B} \sum_{j \in S} \mu(j) \cdot v_{i,j}\right)
    \enspace.
\end{align*}

For every bidder $i \in B$, let $\Sigma_i = \sum_{j \in I} \mu(j) \cdot v_{ij}$. It is easy to see that $v(I) = \smax_{i \in B} \Sigma_i$ (in other words, the second highest $\Sigma_i$ value is $v(I)$). Let $i^* \in B$ be the index maximizing $\Sigma_{i^*}$ (breaking ties arbitrary). Consider a set $S \in \cappart_{t = 1}^m \cP'_t$. The elements of $S$ contribute at least $\smax_{i \in B} \sum_{j \in S} \mu(j) \cdot v_{i,j}$ to at least two of the values: $\Sigma_1, \ldots, \Sigma_n$. Thus, they contribute at least the same quantity to the sum $\sum_{i \in B \setminus \{i^*\}} \Sigma_i$. This means that at least one of the values $\{\Sigma_i\}_{i \in B \setminus \{i^*\}}$ must be at least:
\[
	\frac{\sum_{S \in \cappart_{t = 1}^m \cP'_t} \left(\smax_{i \in B} \sum_{j \in S} \mu(j) \cdot v_{i, j}\right)}{k - 1}
	=
	\frac{R(\cappart_{t = 1}^m \cP'_t)}{k - 1} \enspace.
\]
By definition $\Sigma_{i^*}$ must also be at least that large, and therefore, \[R(\{I\}) = v(I) \geq R(\cappart_{t = 1}^m \cP'_t) / (k - 1). \qedhere\]\end{itemize}
\end{proof}

\subsubsection{A $5$-approximation algorithm for Local Expert mediators} \label{sec_alg_local_experts}

In this subsection we consider an interesting special case of {\DSP} which is henceforth shown to admit a constant factor approximation.

\begin{definition}[Local Expert mediators]
A mediator $t$ in a {\DSP} instance is a \emph{local experts} if there exists a set $I_t \subseteq I$ such that: 
$\cP_t = \{\{j\} \mid j \in I_t\} \cup \{I \setminus I_t\}$.
\end{definition}

Informally, a local expert mediator has perfect 
knowledge about a single set $I_t$---if the item belongs to $I_t$, he can tell exactly which item it is. Our objective in the rest of the section is to prove Theorem~\ref{thm:alg_local_experts}, i.e., to describe a $5$-approximation algorithm for instances of {\DSP} consisting of only local expert mediators.

We begin the proof with an upper bound on the revenue of the optimal joint strategy, which we denote by $\cP^*$. To describe this bound, we need some notation. We use $\hat{I}$ to denote the set of items that are within the experty field of some mediator (formally, $\hat{I} = \bigcup_{t \in M} I_t$). Additionally, for every item $j \in I$, $h_j$ and $s_j$ denote $\mu(j)$ times the largest value and second largest value, respectively, of $j$ for any bidder (more formally, $h_j = \mu(j) \cdot \max_{i \in B} v_{i,j}$ and $s_j = \mu(j) \cdot \smax_{i \in B} v_{i,j}$). 

Next, we need to partition the items into multiple sets. The optimal joint partition $\cP^*$ is obtained from partitions $\{\cP^*_t\}_{t \in M}$, where $\cP^*_t$ is a possible partition for mediator $t$. Each part of $\cP^*$ is the intersection of $|M|$ parts, one from each partition in $\{\cP^*_t\}_{t \in M}$. On the other hand, each part of $\cP^*_t$ is a subset of $I_t$, except for maybe a single part. Hence, there exists at most a single part $I_0 \in \cP^*$ such that $I_0 \not \subseteq I_t$ for any $t \in M$. For ease of notation, if there is no such part (which can happen when $\hat{I} = I$) we denote $I_0 = \varnothing$. To partition the items of $I \setminus I_0$, we associate each part $S \in \cP^* \setminus \{I_0\}$ with an arbitrary mediator $t$ such that $S \subseteq I_t$, and denote by $A_t$ the set of items of all the parts associated with mediator $t$. Observe that the construction of $A_t$ guarantees that $A_t \subseteq I_t$. Additionally, $\{I_0\} \cup \{A_t\}_{t \in M}$ is a disjoint partition of $I$.

A different partition of the items partitions them according to the bidder that values them the most. In other words, for every $1 \leq i \leq k$, $H_i$ is the set of items for which bidder $i$ has the largest value. If multiple bidders have the same largest value for an item, we assign it to the set $H_i$ of an arbitrary one of these bidders. Notice that the construction of $H_i$ guarantees that the sets $\{H_i\}_{i \in B}$ are disjoint.

Finally, for every set $S \subseteq I$, we use $\phi(S)$ to denote the sum of the $|B| - 1$ smaller values in $\{\sum_{j \in H_i \cap S} h_j\}_{i \in B}$, i.e., the sum of all the values except the largest one. In other words, we calculate for every bidder $i$ the sum of its values for items in $H_i \cap S$, and then add up the $|B| - 1$ smaller sums. Using all the above notation we can now state our promised upper bound on $R(\cP^*)$.

\begin{lemma} \label{lem:upper_bound}
$R(\cP^*) \leq \mu(I_0) \cdot v(I_0) + \sum_{j \in \hat{I}} s_j + \sum_{t \in M} \phi(A_t)$.
\end{lemma}
\begin{proof}
Fix an arbitrary mediator $t \in M$, and let $i$ be the bidder whose term is not counted by $\phi(A_t)$. For every part $S \in \cP^*$ associated with $t$, let $i'$ be a bidder other than $i$ that has one of the two largest bids for $S$. By definition:
\[
	\mu(S) \cdot v(S)
	=
	\smax_{i'' \in B} \sum_{j \in S} \mu(j) \cdot v_{i'',j}
	\leq
	\sum_{j \in S} \mu(j) \cdot v_{i',j}
	\leq
	\sum_{j \in S \cap H_i} s_j + \sum_{j \in S \setminus H_i} h_j
	\enspace.
\]
Summing over all parts associated with $t$, we get:
\[
	\sum_{\substack{S \in \cP^*\\S \subseteq A_t}} \mu(S) \cdot v(S)
	\leq
	\sum_{j \in A_t \cap H_i} s_j + \sum_{j \in A_t \setminus H_i} h_j
	\leq
	\sum_{j \in A_t} s_j + \phi(A_t)
	\enspace.
\]
Summing over all mediators, we get:
\[
	R(\cP^*) - \mu(I_0) \cdot v(I_0)
	\leq
	\sum_{t \in M} \left(\sum_{j \in A_t} s_j + \phi(A_t)\right)
	\leq
	\sum_{j \in \hat{I}} s_j + \sum_{t \in M} \phi(A_t)
	\enspace.
	\qedhere
\]
\end{proof}

Our next step is to describe joint partitions that can be found efficiently and upper bound the different terms in the bound given by Lemma~\ref{lem:upper_bound} (up to a constant factor). Finding such partitions for the first two terms is quite straightforward.

\begin{observation} \label{obs:silent_value}
The joint partitions where all mediators are silent $\{I\} = \cappart_{i \in B} \{I\}$ obeys:
\[
	R(\{I\})
	\geq
	\mu(I_0) \cdot v(I_0)
	\enspace.
\]
\end{observation}
\begin{proof}
\[
	R(\{I\})
	=
	\smax_{i \in B} \left(\sum_{j \in I} \mu(j) \cdot v_{i,j}\right)
	\geq
	\smax_{i \in B} \left(\sum_{j \in I_0} \mu(j) \cdot  v_{i,j}\right)
	=
	\mu(I_0) \cdot v(I_0)
	\enspace.
	\qedhere
\]
\end{proof}

\begin{observation} \label{obs:all_talk_value}
The joint partitions $\cP_S = \cappart_{t \in M} \cP_t$ where every mediators reports all his information obeys:
\[
	R(\cP_S)
	=
	R(\{\{j\}_{j \in \hat{I}}\} \cup \{I \setminus \hat{I}\})
	\geq
	\sum_{j \in \hat{I}} \mu(j) \cdot \smax_{i \in B} v_{i,j}
	=
	\sum_{j \in \hat{I}} s_j
	\enspace.
\]
\end{observation}

It remains to find a joint partition that upper bounds, up to a constant factor, the third term in the bound given by Lemma~\ref{lem:upper_bound}. If one knows the sets $\{A_t\}_{t \in M}$, then one can easily get such a partition using the method of Ghosh et al.~\cite{GNS07}. In this method, one partitions every set $A_t$ into the parts $\{A_t \cap H_i\}_{i = 1}^t$ and sort these parts according to the value of $\sum_{j \in A_t \cap H_i} h_j$. Then, with probability $\nicefrac{1}{2}$ every even part is united with the part that appears after it in the above order, and with probability $\nicefrac{1}{2}$ it is united with the part that appears before it in this order. It is not difficult to verify that if the part of bidder $i$ is not the first in the order, then with probability $\nicefrac{1}{2}$ it is unified with the part that appears before it in the order, and then it contributes $\sum_{j \in A_t \cap H_i} h_j$ to the revenue. Hence, the expected contribution to the revenue of the parts produced from $A_t$ is at least $\nicefrac{1}{2} \cdot \phi(A_t)$. 

Algorithm~\ref{alg:local_experts} can find a partition that is competitive against $\sum_{t \in M} \phi(A_t)$ without knowing the sets $\{A_t\}_{t \in M}$. The algorithm uses the notation of a \emph{cover}. We say that a set $S_j$ is a cover of an element $j \in I_t \cap H_i$ if $S_j \subseteq I_t \cap H_{i'}$ for some $i \neq i'$.

\begin{algorithm}
\caption{\textsf{Local Experts - Auxiliary Algorithm}} \label{alg:local_experts}
Let $I' \gets \hat{I}$ and $\cP \gets \{I \setminus \hat{I}\}$.\\
\While{$I' \neq \varnothing$}{
	Let $j$ be the element maximizing $h_j$ in $I'$.\\
	Find a cover $S_j \subseteq I'$ of $j$ obeying $h_j \leq \sum_{j' \in S_j} h_{j'} \leq 2h_j$, or maximizing $\sum_{j' \in S_j} h_{j'}$ if no cover of $j$ makes this expression at least $h_j$. \label{line:cover_finding}\\
	Add the part $S_j \cup \{j\}$ to $\cP$, and remove the elements of $S_j \cup \{j\}$ from $I'$.
}
\Return{$\cP$}
\end{algorithm}

Notice that the definition of cover guarantees that a part containing both $j$ and $S_j$ contributes to the revenue at least $\min\{h_j, \sum_{j' \in S_j h_{j'}}\}$. Using this observation, each iteration of Algorithm~\ref{alg:local_experts} can be viewed as trying to extract revenue from element $j$. Additionally, observe that the partition $\cP$ produced by Algorithm~\ref{alg:local_experts} can be presented as a joint partition since every part in it, except for $I \setminus \hat{I}$, contains only items that belong to a single set $I_t$ (for some mediator $t \in M$).

\begin{observation}
Algorithm~\ref{alg:local_experts} can be implemented in polynomial time.
\end{observation}
\begin{proof}
One can find a cover $S_j$ maximizing $\sum_{j' \in S_j} h_{j'}$ in line~\ref{line:cover_finding} of the algorithm by considering the set $I_t \cap H_{i'} \cap I'$ for every mediators $t$ and bidder $i'$ obeying $j \in I_t$ and $j \not \in H_{i'}$. Moreover, if this cover is of size larger than $2h_j$, then by removing elements from this cover one by one the algorithm must find a cover $S'_j$ obeying $h_j \leq \sum_{j' \in S'_j} h_{j'} \leq 2h_j$ because $j$ is the element maximizing $h_j$ in $I'$.
\end{proof}

The following lemma relates the revenue of the set produced by Algorithm~\ref{alg:local_experts} to $\sum_{t \in M} \phi(A_t)$.
\begin{lemma} \label{lem:greedy_step}
No iteration of the loop of Algorithm~\ref{alg:local_experts} decreases the expression $R(\cP) + \nicefrac{1}{3} \cdot \sum_{t \in M} \phi(A_t \cap I')$.\footnote{Before the algorithm terminates $\cP$ is a partial partition in the sense that some items do not belong to any part in it. However, the definition of $R(\cP)$ naturally extends to such partial partitions.}
\end{lemma}
\begin{proof}
Fix an arbitrary iteration. There are two cases to consider. First, assume $h_j \leq \sum_{j' \in S_j} h_{j'} \leq 2h_j$. In this case the increase in $R(\cP)$ during this iteration is:
\[
	\mu(S_j \cup \{j\}) \cdot v(S_j \cup \{j\})
	\geq
	\min\left\{h_j, \sum_{j' \in S_j} h_{j'}\right\}
	=
	h_j
	\enspace.
\]
On the other hand, one can observe that, when removing an element $j'$ from $S$, the value of $\phi(S)$ can decrease by at most $h_{j'}$. Hence, the decrease in $\sum_{t \in M} \phi(A_t \cap I')$ during this iteration can be upper bounded by:
\[
	h_j + \sum_{j' \in S_j} h_{j'}
	\leq
	3h_j
	\enspace.
\]

Consider now the case $\sum_{j' \in S_j} h_{j'} < h_j$. In this case the increase in $R(\cP)$ during the iteration is:
\[
	\mu(S_j \cup \{j\}) \cdot v(S_j \cup \{j\})
	\geq
	\min\left\{h_j, \sum_{j' \in S_j} h_{j'}\right\}
	=
	\sum_{j' \in S_j} h_{j'}
	\enspace.
\]
If $j$ does not belong to $A_t$ for any mediator $t$, then by the above argument we can bound the decrease in $\sum_{t \in M} \phi(A_t \cap I')$ by $\sum_{j' \in S_j} h_{j'}$. Hence, assume from now on that there exists a mediator $t'$ and a bidder $i$ such that $j \in A_{t'} \cap H_i$.
Let $i' \neq i$ be a bidder maximizing $\sum_{j' \in H_{i'} \cap A_{t'} \cap I'} h_{j'}$. Clearly, the removal of a single element from $I'$ can decrease $\phi(A_{t'} \cap I')$ by no more than $\sum_{j' \in H_{i'} \cap A_{t'} \cap I'} h_{j'}$. Hence, the decrease in $\sum_{t \in M} \phi(A_t \cap I')$ during the iteration of the algorithm can be upper bounded by:
\[
	\sum_{j' \in H_{i'} \cap A_{t'} \cap I'} h_{j'} + \sum_{j' \in S_j} h_{j'}
	\enspace.
\]
On the other hand, $H_{i'} \cap A_{t'} \cap I'$ is a possible cover for $j$, and thus, by the optimality of $S_j$:
\[
	\sum_{j' \in H_{i'} \cap A_{t'} \cap I'} h_{j'}
	\leq
	\sum_{j' \in S_j} h_{j'}
	\enspace.
	\qedhere
\]
\end{proof}

\begin{corollary} \label{cor:alg_value}
$R(\cP_A) \geq \nicefrac{1}{3} \cdot \sum_{t \in M} \phi(A_t)$, where $\cP_A$ is the partition produced by Algorithm~\ref{alg:local_experts}.
\end{corollary}
\begin{proof}
After the initialization step of Algorithm~\ref{alg:local_experts} we have:
\[
	R(\cP) + \nicefrac{1}{3} \cdot \sum_{t \in M} \phi(A_t \cap I')
	\geq
	\nicefrac{1}{3} \cdot \sum_{t \in M} \phi(A_t)
	\enspace.
\]
On the other hand, when the algorithm terminates:
\[
	R(\cP) + \nicefrac{1}{3} \cdot \sum_{t \in M} \phi(A_t \cap I')
	=
	R(\cP_A)
\]
because $I' = \varnothing$. The corollary now follows from Lemma~\ref{lem:greedy_step}.
\end{proof}

We are now ready to prove Theorem~\ref{thm:alg_local_experts}.

\begin{reptheorem}{thm:alg_local_experts}
If mediators are local experts, 
there exists an efficient $5$-approx\-imation algorithm for {\DSP}.
\end{reptheorem}
\begin{proof}
Consider an algorithm that outputs the best solution out of $\{I\}$, $\cP_S$ and $\cP_A$. The following inequality shows that at least one of these joint partitions has a revenue of $R(\cP^*)/5$:
\[
	R(\{I\}) + R(\cP_S) + 3R(\cP_A)
	\geq
	\sum_{j \in \hat{I}} s_j + \mu(I_0) \cdot v(I_0) + \sum_{t \in M} \phi(A_t)
	\geq
	R(\cP^*)
	\enspace,
\]
where the first inequality holds by Observations~\ref{obs:silent_value} and~\ref{obs:all_talk_value} and Corollary~\ref{cor:alg_value}; and the second inequality uses the upper bound on $R(\cP^*)$ proved by Lemma~\ref{lem:upper_bound}.
\end{proof}
	
} 

%% file: shapley.tex
\section{The Strategic Problem}

This section explores the {\DSP} problem from a strategic viewpoint, 
in which the auctioneer \emph{cannot} control the signals produced by each mediator, and is, therefore, trying to solicit information from the
mediators that would yield a maximal revenue in the auction.
In other words, the objective of the auctioneer is to design a mechanism  
$\cM$ whose equilibria (i.e., the signals $\cP_1' ,\cP_2' ,\ldots , \cP_m '$ which are now chosen strategically by the mediators) 
induce maximum revenue. Our first contribution is the introduction of the \emph{Shapley mechanism}, whose definition appears in Subsection~\ref{sec_shapley}. Subsection~\ref{sec_shapley} also proves some interesting properties of the Shapley mechanism (Theorems~\ref{thm_pure_eq} and~\ref{th:better_then_silence}). Subsection~\ref{sec_poa} studies the price of anarchy and price of stability of the {\DSP} game induced by the Shapley mechanism (Theorems~\ref{thm_poa_ub_combined} and~\ref{thm_unbounded_pos}). Finally, Subsection~\ref{ssc:inevitable} shows that the Shapley mechanism is the only possible mechanism if one insists on a few natural requirements.

\subsection{The Shapley Mechanism}\label{sec_shapley}

In this subsection we describe a mechanism $\cS$ which determines the payments to the mediators as a function of the reported signals. 
Our mechanism aims to incentivize mediators to report 
useful information, with the hope that global efficiency emerges despite selfish behavior of each mediator. In the remainder of the paper 
we study the mechanism $\cS$ and the game $\DSP_\cS$ it induces. For the sake of generality, we describe $\cS$ for a game generalizing {\DSP}.

Consider a game $\cG_m$ of $m$ players where each player $t$ has a finite set $A_t$ of possible strategies, one of which $\varnothing_t \in A_t$ is called the null strategy of $t$. The value of a strategy profile in the game $\cG_m$ is determined by a value function $v : A_1 \times A_2 \times \ldots \times A_m \rightarrow \mathbb{R}$. A mechanism $M = (\Pi_1, \Pi_2, \ldots, \Pi_m)$ for $\cG_m$ is a set of payments rules. In other words, if the players choose strategies $a_1 \in A_1, a_2 \in A_2, \ldots, a_m \in A_m$, then the payment to player $t$ under mechanism $M$ is $\Pi_t(v, a_1, a_2, \ldots, a_m)$. Notice that $\DSP$ fits the definition of $\cG_m$ when $A_t = \Omega(\cP_t)$ is the set of partitions that $t$ can report for every mediator $t$, and $\varnothing_t$ is the silence strategy $\{I\}$. The appropriate value function $v$ for $\DSP$ is the function $R(\times_{t = 1}^m \cP'_t)$, where $\cP'_t \in A_t$ is the strategy of mediator $t$. In other words, the value function $v$ of a $\DSP$ game is equal to the revenue of the auctioneer.

Given a strategy profile $a = (a_1,a_2,\dotsc, a_m)$, and subset $J\in [m]$ of players, we write $a_J$ to denote a strategy profiles where the players of $J$ play their strategy in $a$, and the other players play their null strategies. We abuse notation and denote by $\varnothing$ the strategy profile $a_{\varnothing}$ where all players play their null strategies. Additionally, we write $(a'_t,a_{-t})$ to denote a strategy profile where player $t$ plays $a'_t$ and the rest of the players follow the strategy profile $a$. 
The mechanism $\cS$ we propose distributes
the increase in the value of the game (compared to $v(\varnothing)$) among the players according to their Shapley value:
it pays each player his expected marginal contribution to the value according to a uniformly random ordering of the $m$ player. Formally, the payoff for player $t$ given a strategy profile $a$ is
\begin{equation}\label{def_shapley_val_1}
 \Pi_t(a) = \frac{1}{m!} \cdot \sum_{\sigma \in S_m} \left[v\left(a_{\{\sigma^{-1}(j) \mid 1 \leq j \leq \sigma(t)\}}\right) - v\left(a_{\{\sigma^{-1}(j) \mid 1 \leq j < \sigma(t)\}}\right)\right] \enspace,
\end{equation}
which can alternatively be written as
\begin{equation}\label{def_shapley_val_2}
\Pi_t(a)  =  \sum_{J\subseteq [m] \setminus \{t\}} \gamma_J \left(v(a_{J \cup \{t\}}) - v(a_{J})\right) \enspace,
\end{equation}
where $\gamma_J = \frac{|J|! (m-|J|-1)!}{m!}$ is the probability that the players of $J$ appear before player $t$
when the players are ordered according to a uniformly random permutation $\sigma \in_R S_m$. We use both definitions~\eqref{def_shapley_val_1} and~\eqref{def_shapley_val_2}
interchangeably, as each one is more convenient in some cases than the other.

Clearly, the mechanism $\cS$ is anonymous (symmetric). 
The main feature of the Shapley mechanism is that it is efficient. In other words, the sum of the payoffs is exactly equal to the total increase in value
(in the case of {\DSP}, the surplus revenue of the auctioneer compared to the initial state):
\begin{proposition}[Efficiency property]\label{prop_revenue_payoffs}
For every strategy profile $a = (a_1, a_2, \ldots, a_m)$, \[ v(a) - v(\varnothing) = 
\sum_{t = 1}^m \Pi_t(a) \enspace.\]
\end{proposition}

\begin{proof}
Recall that the payoff of mediator $t$ is:
\[
  \frac{1}{m!} \cdot \sum_{\sigma \in S_m} \left[v\left(a_{\{\sigma^{-1}(j) \mid 1 \leq j \leq \sigma(t)\}}\right) - v\left(a_{\{\sigma^{-1}(j) \mid 1 \leq j < \sigma(t)\}}\right)\right]
	\enspace.
\]

Summing over all mediators, we get:
\begin{align*}
    \sum_{t = 1}^m \Pi_t(\cP'_t,\cP'_{-t}) =&
    \sum_{t = 1}^m \left\{\frac{1}{m!} \cdot \sum_{\sigma \in S_m} \left[v\left(a_{\{\sigma^{-1}(j) \mid 1 \leq j \leq \sigma(t)\}}\right) - v\left(a_{\{\sigma^{-1}(j) \mid 1 \leq j < \sigma(t)\}}\right)\right]\right\}\\
    ={} &
		\frac{1}{m!} \cdot \sum_{\sigma \in S_m} \sum_{t = 1}^m \left[v\left(a_{\{\sigma^{-1}(j) \mid 1 \leq j \leq \sigma(t)\}}\right) - v\left(a_{\{\sigma^{-1}(j) \mid 1 \leq j < \sigma(t)\}}\right)\right]\\
    ={} &
    \frac{1}{m!} \cdot \sum_{\sigma \in S_m} \left[v\left(a_{\{\sigma^{-1}(j) \mid 1 \leq j \leq m\}}\right) - v\left(a_{\varnothing}\right)\right]
    =
    v(a) - v(\varnothing)
    \enspace.
		\qedhere
\end{align*}
\end{proof}

Proposition \ref{prop_revenue_payoffs} implies the following theorem. Notice that Theorem~\ref{th:better_then_silence} is in fact a restriction of this theorem to the game $\DSP_\cS$.
\begin{theorem}
For every Nash equilibrium $a$, $v(a) \geq v(\varnothing)$.
\end{theorem}
\begin{proof}
A player always has the option of playing his null strategy, which results in a zero payoff for him. Thus, the payoff of a player in a Nash equilibrium 
can never be negative. Hence, by Proposition~\ref{prop_revenue_payoffs}: $v(a) \geq v(\varnothing) + \sum_{i=1}^m \Pi_t(a) \geq v(\varnothing)$. 
\end{proof}

Next, let us prove Theorem~\ref{thm_pure_eq}. For convenience, we restate it below.

\begin{reptheorem}{thm_pure_eq}
Let $\G_m$ be a non-cooperative $m$-player game in which 
the payoff of each player is set according to $\cS$. Then
$\G_m$ admits a pure Nash equilibrium. Moreover, best response dynamics are guaranteed to converge to such an equilibrium.
\end{reptheorem}
\begin{proof}
We prove the theorem by showing that $\G_m$ is an exact potential game, which in turn implies all the conclusions of the theorem. Recall that an exact potential game is a game for which there exists a potential function $\Phi\colon A_1 \times A_2 \times \dots \times A_t \to \mathbb{R}$ such that every strategy profile $a$ and possible deviation $a'_t \in A_t$ of a player $t$ obey:
\begin{equation} \label{eq:potential}
	\Pi_t(a'_t, a_{-t}) - \Pi_t(a)
	=
	\Phi(a'_t, a_{-t}) - \Phi(a)
	\enspace.
\end{equation}

In our case the potential function is: $\Phi(a) = \sum_{J\subseteq [m]} \beta_J \cdot v(a_J)$, where $\beta_J = \frac{(|J| - 1)! (m-|J|)!}{m!}$. Let us prove that this function obeys~\eqref{eq:potential}. It is useful to denote by $a'$ the strategy profile $(a'_t, a_{-t})$. By definition:
\begin{equation} \label{eq_i_in_s}
		\Pi_t(a') - \Pi_t(a)
		=
    \sum_{J\subseteq [m]\setminus\{i\}} \gamma_J \left[  v(a_{J \cup \{i\}}) - v(a_J) \right]
		-
		\sum_{J\subseteq [m]\setminus\{i\}} \gamma_J \left[  v(a'_{J \cup \{i\}}) - v(a'_J) \right]
    \enspace.
\end{equation}
For $J \subseteq [m]\setminus\{i\}$, we have $a_J = a'_J$. Plugging this observation into \eqref{eq_i_in_s}, and rearranging, we get:
\begin{equation} \label{eq_geq_partial}
    \Pi_t(a') - \Pi_t(a)
		=
    \sum_{J\subseteq [m]\setminus\{i\}} \gamma_J \left[  v(a_{J \cup \{i\}}) - v(a'_{J \cup \{i\}}) \right]
    \enspace.
\end{equation}
For every $J$ containing $i$ we get: $\alpha_{J \setminus \{i\}} = \beta_J$. Using this observation and the previous observation that $a_J = a'_J$ for $J \subseteq [m]\setminus\{i\}$, \eqref{eq_geq_partial} can be replaced by:
\begin{equation*}
		\Pi_t(a') - \Pi_t(a)
		=
    \sum_{J\subseteq [m]} \beta_J (v(a'_J) - v(a_J))
		=
		\Phi(a') - \Phi(a)
    \enspace.
		\qedhere
\end{equation*}
\end{proof}

\noindent Before concluding this section, a few remarks regarding the use of $\cS$ to {\DSP} are in order:
\begin{enumerate}
\item The reader may wonder why the auctioneer cannot impose on the mediators any desired outcome $\times_{t \in M} \cP'_t$ by offering mediator $t$ a
negligible payment if he signals $\cP'_t$, and no payment otherwise. However, implementing such a mechanism requires the auctioneer to know the 
information sets $\cP_t$ of each mediator \emph{in advance}. In contrast, our mechanism requires access only to the outputs of the mediators.
\item Proposition~\ref{prop_revenue_payoffs} implies that the auctioneer distributes the entire surplus among the mediators, which seems to defeat the 
purpose of the mechanism. However, in the target application she can scale the revenue by a factor $\alpha \in (0,1]$ and only distribute the corresponding 
fraction of the surplus. As all of our results are invariant under scaling, this trick can be applied in a black box fashion. Thus, we assume,
without loss of generality, $\alpha=1$.
\item We assume mediators never report a signal which is inconsistent with the true identity of the sold element $j_R$. The main justification for this assumption is that the mediators' signals must be consistent with one another
(as they refer to a single element $j_R$). Thus, given that sufficiently many mediators are honest, 
``cheaters'' can be easily detected. 
\end{enumerate}

%% file: PoA.tex
\subsection{The Price of Anarchy and Price of Stability of the Shapley Mechanism}\label{sec_poa}

In this section we analyze the PoS and PoA of the $\DSP_\cS$ game.
First, we note that the proof of Theorem~\ref{thm:simple_approximations} in Section~\ref{sec_gen_algorithms} shows that $R(\{I\}) \geq \max\{1, \min\{n, k- 1\}\}$. Together with Theorem~\ref{th:better_then_silence}, we get:
\begin{reptheorem}{thm_poa_ub_combined} 
The price of anarchy of $\DSP_\cS(n,k,m)$ is no more than $\max\{1, \min \{k-1, n\}\}$.
\end{reptheorem}
\noindent \textbf{Remark:} The above statement of Theorem~\ref{thm_unbounded_pos} uses a somewhat different notation than its original statement in Section~\ref{sc:our_results}, but both statements are equivalent. The same is true for the statement of Theorem~\ref{thm_unbounded_pos} below.

Naturally, the upper bound given by Theorem~\ref{thm_poa_ub_combined} applies also to the price of stability of $\DSP_\cS$. The rest of this section is devoted to proving Theorem~\ref{thm_unbounded_pos}, which shows that Theorem~\ref{thm_poa_ub_combined} is asymptotically tight.


\begin{reptheorem}{thm_unbounded_pos}
For every $n \geq 1$, there is a $\DSP_\cS(3n + 1, n + 2, 2)$ game for which the price of stability is at least $n$. Moreover, all the mediators in this game are local experts.
\end{reptheorem}

We begin the proof of Theorem~\ref{thm_unbounded_pos} by describing the $\DSP_\cS(3n + 1, n + 2, 2)$ game whose price of stability we bound. For ease of notation, let us denote this game by $\DSP_n$. 
\begin{description}
	\item[items:] The $3n + 1$ items of $\DSP_n$ all have equal probabilities. It is convenient to denoted them by $\{a_\ell\}_{\ell = 1}^n$, $\{b_\ell\}_{\ell = 1}^n$, $\{c_\ell\}_{\ell = 1}^n$ and $d$.
	\item[bidders:] The $n + 2$ bidders of $\DSP_n$ can be partitioned into $3$ types. One bidder, denoted by $i_G$, has a bid of $\eps$ for the items of $\{b_\ell\}_{\ell = 1}^n$ and a bid of $1$ for all other items, where $\eps \in (0, 1)$ is a value that will be defined later. One bidder, denoted by $i_O$ has a bid of $1$ for item $d$ and a bid of $0$ for all other items. Finally, the other $n$ bidders are denoted by $\{i_\ell\}_{\ell = 1}^n$. Each bidder $i_\ell$ has a bid of $1$ for item $b_\ell$ and a bid of $0$ for all other items.
	\item[mediators:] The two mediators of $\DSP_n$ are denoted by $t_1$ and $t_2$. Both mediators are local experts whose partitions are defined by the sets $\cI_1 = \{a_\ell, b_\ell\}_{\ell = 1}^n$ and $\cI_2 = \{b_\ell, c_\ell\}_{\ell = 1}^n$, respectively.
\end{description}

A graphical sketch of $\DSP_n$ is given by Figure~\ref{fig:dsp_n}. Intuitively, getting a high revenue in $\DSP_n$ requires pairing $b$ items with $a$ or $c$ items. Unfortunately, one mediator can pair the $b$ items with $a$ items, and the other mediator can pair them with $c$ items, thus, creating ``tension'' between the mediators. The main idea of the proof is to show that both mediators are incentivized to report partitions pairing the $b$ items, which results in a joint partition isolating all the $b$ items.

\newcommand\overmat[2]{\makebox[0pt][l]{$\smash{\color{white}\hspace{-0.4cm}\overbrace{\hspace{0.4cm}\phantom{\begin{matrix}#2\end{matrix}}\hspace{0.4cm}}^{\text{\textcolor{black}{#1}}}}$}#2}
\newcommand\bovermat[2]{\makebox[0pt][l]{$\smash{\overbrace{\phantom{\begin{matrix}#2\end{matrix}}}^{\text{#1}}}$}#2}
\setcounter{MaxMatrixCols}{20}
\newcommand\myBox{\makebox[0pt][l]{\smash{\hspace{-0.12cm}\raisebox{-0.08cm}{\text{\scalebox{1.5}{$\Box$}}}}}\phantom{1}}
\newcommand\pBox{\phantom{\myBox}}
\newcommand\myDots{\text{$\cdots$}}
\newcommand\pDots{\phantom{\myDots}}
\begin{figure}
\[
V \;\; = \;\;
\left(
\begin{matrix}
	\bovermat{$a$ items}{1 & 1 & \ldots & 1} & \  &			\bovermat{$b$ items}{\eps & \eps & \ldots & \eps} & \  &			\bovermat{$c$ items}{1 & 1 & \ldots & 1} &	\  &			\overmat{$d$}{1} \\
	0 & 0 & \cdots & 0 &&																		0 & 0 & \cdots & 0 &&																						0 & 0 & \cdots & 0 &&																		1 \\
	0 & 0 & \cdots & 0 &&																		1 & 0 & \cdots & 0 &&																						0 & 0 & \cdots & 0 &&																		0 \\
	0 & 0 & \cdots & 0 &&																		0 & 1 & \cdots & 0 &&																						0 & 0 & \cdots & 0 &&																		0 \\
	\vdots & \vdots && \vdots &&														\vdots & \vdots && \vdots &&																		\vdots & \vdots && \vdots &&														0 \\
	0 & 0 & \cdots & 0 &&																		0 & 0 & \cdots & 1 &&																						0 & 0 & \cdots & 0 &&																		0 \\
\end{matrix}
\right)
\begin{array}{cc}
	\leftarrow & i_G \\
	\leftarrow & i_O \\
	\leftarrow & i_1 \\
	\leftarrow & i_2 \\
	& \vdots \\
	\leftarrow & i_n \\
\end{array}
\]
\[
	\cP_1 \;\; = \;\;
	\begin{matrix}
		\myBox & \myBox & \myDots & \myBox & \ & \myBox & \myBox & \myDots & \myBox & \ & \pBox & \pBox & \pDots & \pBox & \ & \pBox \\
	\end{matrix}
	\mspace{78mu}
\]
\[
	\cP_2 \;\; = \;\;
	\begin{matrix}
		\pBox & \pBox & \pDots & \pBox & \ & \myBox & \myBox & \myDots & \myBox & \ & \myBox & \myBox & \myDots & \myBox & \ & \pBox \\
	\end{matrix}
	\mspace{78mu}
\]
\caption{A graphical representation of $\DSP_n$. $\cP_1$ and $\cP_2$ are the partitions of mediators $t_1$ and $t_2$, respectively.} \label{fig:dsp_n}
\end{figure}

The following observation simplifies many of our proofs.

\begin{observation} \label{obs:values}
The contribution $\mu(S) \cdot v(S)$ of a part $S$ to the revenue $R(\cP)$ of a partition $\cP \ni S$ is:
\begin{compactitem}
	\item $(3n + 1)^{-1}$ if $S$ contains $d$.
	\item $(3n + 1)^{-1}$ if $S$ contains an item of $\{b_\ell\}_{\ell = 1}^n$ and $|S| \geq 2$.
	\item $\eps / (3n + 1)$ if $S$ contains only a single item, and this item belongs to $\{b_\ell\}_{\ell = 1}^n$.
	\item $0$ otherwise.
\end{compactitem}
\end{observation}
\begin{proof}
Notice that the only bidder that values more than one item is $i_G$, and thus, $i_G$ is the single bidder that can have a bid larger than $1/|S|$ for $S$. Hence, $v(S)$ is upper bounded by $1/|S|$, and one can bound the contribution of $S$ by:
\[
	\mu(S) \cdot v(S)
	\leq
	\frac{|S|}{3n + 1} \cdot \frac{1}{|S|}
	=
	(3n + 1)^{-1}
	\enspace.
\]
If the part $S$ contains the item $d$, then both bidders $i_G$ and $i_O$ have bids of at least $1/|S|$ for it. Hence, the above upper bound on the contribution of $S$ is in fact tight. Similarly when $S$ contains an item $b_\ell$ and some other item $j$, it has a contribution of $(3n + 1)^{-1}$ since two bidders have a bid of at least $1/|S|$ for it, bidder $i_{\ell}$ and a second bidder that depends on $j$:
\begin{compactitem}
	\item If $j = b_{\ell'}$ for some $\ell' \neq \ell$, bidder $i_{\ell'}$.
	\item If $j = d$, $j = a_{\ell'}$ or $j = c_{\ell'}$, bidder $i_G$.
\end{compactitem}

If the part $S$ contains only a single item $b_\ell$, then it gets none zero bids only from two bidders: a bid of $1$ from $i_\ell$ and a bid of $\eps$ from bidder $i_G$. Hence, $\mu(S) \cdot v(S) = (3n + 1)^{-1} \cdot \eps$. Finally, if $S$ does not fall into any of the cases considered above, then it must contain only items of $\{a_\ell, c_\ell\}_{\ell = 1}^n$. Such a part receives a non-zero bid only from bidder $i_G$, and thus, $v(S) = 0$.
\end{proof}

The next lemma gives a lower bound on the optimal revenue of $\DSP_n$.

\begin{lemma} \label{lem:opt_revenue}
The optimal revenue of $\DSP_n$ is at least $(n + 1)/(3n + 1)$.
\end{lemma}
\begin{proof}
Consider a scenario in which $t_1$ pairs the $b$ items with $a$ items, and $t_2$ is silent. Formally, the two mediators use the following strategies $O_1 = \{\{a_\ell, b_\ell\}\}_{\ell = 1}^n \cup \{\{c_\ell\}_{\ell = 1}^n \cup \{d\}\}$ and $O_2 = \{I\}$, respectively. Observe that these are indeed feasible strategies for $t_1$ and $t_2$, respectively. By Observation~\ref{obs:values}, the revenue of $\DSP_n$ given these strategies is:
\begin{align*}
	R(O_1 \cappart O_2)
	={} &
	R(O_1)
	=
	\sum_{\ell = 1}^n \mu(\{a_\ell, b_\ell\}) \cdot v(\{a_\ell, b_\ell\}) + \mu(\{c_\ell\}_{\ell = 1}^n \cup \{d\}) \cdot v(\{c_\ell\}_{\ell = 1}^n \cup \{d\})\\
	={} &
	\sum_{\ell = 1}^n \frac{1}{3n + 1} + \frac{1}{3n + 1}
	=
	\frac{n + 1}{3n + 1}
	\enspace.
	\qedhere
\end{align*}
\end{proof}

Our next objective is to get an upper bound on the revenue of any Nash equilibrium of $\DSP_n$. We say that an item $a_\ell$ is \emph{redundant} in a strategy $\cP'_1$ of $t_1$ if the part $S \in \cP'_1$ containing $a_\ell$ obeys one of the following:
\begin{compactitem}
	\item $S$ contains the item $d$ or an item $a_{\ell'}$ for some $\ell' \neq \ell$..
	\item $S$ contains no items of $\{b_\ell\}_{\ell = 1}^n$.
\end{compactitem}
The next lemma shows that a redundant item is indeed redundant in the sense that removing it does not change the contribution to the revenue of parts containing it.

\begin{lemma} \label{lem:redundent_no_loss}
If $a_\ell$ is \emph{redundant} in a strategy $\cP'_1$ of $t_1$ and $S \in \cP'_1$ is the part containing $a_\ell$, then $\mu(S) \cdot v(S) = \mu(S \setminus \{a_\ell\}) \cdot v(S \setminus \{a_\ell\})$. Moreover, for every possible strategy $\cP'_2$ of $t_2$, if $S' \in \cP'_1 \cappart \cP'_2$ is the part of $\cP'_1 \cappart \cP'_2$ containing $a_\ell$, then $\mu(S') \cdot v(S') = \mu(S' \setminus \{a_\ell\}) \cdot v(S' \setminus \{a_\ell\})$
\end{lemma}
\begin{proof}
We prove the second part of the lemma. The first part follows from it since one possible choice of $\cP'_2$ is $\{I\}$. Since $a_\ell$ is redundant in $P_1$, one of the following three cases must hold:
\begin{itemize}
	\item The first case is when $d \in S$. Observe that $d$ must share a part with $a_\ell$ in $\cP'_2$, and thus, $d$ belongs also to $S'$. Hence, by Observation~\ref{obs:values}, $\mu(S') \cdot v(S') = (3n + 1)^{-1} = \mu(S' \setminus \{a_\ell\}) \cdot v(S' \setminus \{a_\ell\})$.
	\item The second case is when $(\{b_\ell\}_{\ell = 1}^n \cup \{d\}) \cap S = \varnothing$. Since $S'$ is a subset of $S$, we get also $(\{b_\ell\}_{\ell = 1}^n \cup \{d\}) \cap S' = \varnothing$, and thus, by Observation~\ref{obs:values}, $\mu(S') \cdot v(S') = 0 = \mu(S' \setminus \{a_\ell\}) \cdot v(S' \setminus \{a_\ell\})$.
	\item The third case is when there exits $\ell' \neq \ell$ for which $a_{\ell'} \in S$. Since $a_\ell$ and $a_{\ell'}$ must share a part in $\cP_2$, they both belong also to $S'$. Using Observation~\ref{obs:values} one can verify that removing $a_\ell$ from a set $S'$ containing $a_{\ell'}$ can never change $\mu(S') \cdot v(S')$. \qedhere
\end{itemize}
\end{proof}

A strategy $\cP'_1$ of $t_1$ is called \emph{$a$-helped} if every part in it that contains an item of $\{b_\ell\}_{\ell = 1}^n$ contains also an item of $\{a_\ell\}_{\ell = 1}^n$. The next lemma shows that every strategy of $t_1$ is dominated by an $a$-helped one.
\begin{lemma} \label{lem:make_helper}
If $\cP'_1$ is a strategy of $t_1$, then there exists an $a$-helped strategy $\cP''_1$ of $t_1$ such that: $\Pi_1(\cP''_1, \cP'_2) \geq \Pi_1(\cP'_1, \cP'_2)$ for every strategy $\cP'_2$ of $t_2$. Moreover, for every pair of items from $\{b_\ell\}_{\ell = 1}^n \cup \{d\}$, $\cP''_1$ separates this pair (i.e., each item of the pair appear in a different part of $\cP''_1$) if and only if $\cP'_1$ does.
\end{lemma}
\begin{proof}
Let $D(\cP'_1)$ be the number of parts in $\cP'_1$ that contain an item of $\{b_\ell\}_{\ell = 1}^n$ but no items of $\{a_\ell\}_{\ell = 1}^n$. We prove the lemma by induction on $D(\cP'_1)$. If $D(\cP'_1) = 0$, then $\cP'_1$ is $a$-helped and we are done. It remains to prove the lemma for $D(\cP'_1) > 0$ assuming that it holds for every strategy $\hat{\cP}'_1$ for which $D(\cP'_1) > D(\hat{\cP}'_1)$.

Since $D(\cP'_1) > 0$ and the number of $\{a_\ell\}_{\ell = 1}^n$ items and $\{b_\ell\}_{\ell = 1}^n$ items is equal, there must be in $\cP'_1$ either a part that contains at least two items of $\{a_\ell\}_{\ell = 1}^n$ or a part that contains an item of $\{a_\ell\}_{\ell = 1}^n$ and no items of $\{b_\ell\}_{\ell = 1}^n$. Both options imply the existence of a redundant item $a_{\ell'}$ in $\cP'_1$. Additionally, let $S$ be one of the parts of $\cP'_1$ counted by $D(\cP'_1)$, i.e., a part that contains an item of $\{b_\ell\}_{\ell = 1}^n$ but no items of $\{a_\ell\}_{\ell = 1}^n$.

Consider the strategy $\tilde{\cP}'_1$ obtained from $\cP'_1$ by moving $a_{\ell'}$ to the part $S$. By Lemma~\ref{lem:redundent_no_loss} the contribution to the revenue of the part containing $a_{\ell'}$ in $\cP'_1$ does not decrease following the removal of $a_{\ell'}$ from it. On the other hand, Observation~\ref{obs:values} implies that adding items to a part can only increase its contribution. Hence, the contribution of $S$ to the revenue is at least as large in $\tilde{\cP}'_1$ as in $\cP'_1$. Combining both arguments, we get: $R(\tilde{\cP}'_1) \geq R(\cP'_1)$. An analogous argument can be used to show also that $R(\tilde{\cP}'_1 \cappart \cP'_2) \geq R(\cP'_1 \cappart \cP'_2)$ for every strategy $\cP'_2$ of $t_2$. Thus, by definition, $\Pi_1(\tilde{\cP}'_1, \cP'_2) \geq \Pi_1(\cP'_1, \cP'_2)$.

Observe that the construction of $\tilde{\cP}'_1$ from $\cP'_1$ guarantees that $D(\tilde{\cP}'_1) = D(\cP'_1) - 1$. Hence, by applying the induction hypothesis to $\tilde{\cP}'_1$ we get a strategy $\cP''_1$ which, for every strategy $\cP'_2$ of $t_2$, obeys the inequality:
\[
	\Pi_1(\cP''_1, \cP'_2)
	\geq
	\Pi_1(\tilde{\cP}'_1, \cP'_2)
	\geq
	\Pi_1(\cP'_1, \cP'_2)
\]
Moreover, the construction of $\tilde{\cP}'_1$ from $\cP'_1$ does not move any items of $\{b_\ell\}_{\ell = 1}^n \cup \{d\}$, thus, $\tilde{\cP}'_1$ separates pairs from this set if and only if $\cP'_1$ does. The lemma now follows since the induction hypothesis guarantees that $\cP''_1$ separates pairs from the above set if and only if $\tilde{\cP}'_1$ does.
\end{proof}

Using the previous lemma we can prove an important property of not strictly dominated strategies of $t_1$.

\begin{lemma} \label{lem:isolation}
If $\cP'_1$ is not a strictly dominated strategy of $t_1$, then $\cP'_1$ isolates the items of $\{b_\ell\}_{\ell = 1}^n \cup \{d\}$ from each other.
\end{lemma}
\begin{proof}
Assume towards a contradiction that the lemma does not hold, and let $\cP'_1$ be a counter example. In other words, $\cP'_1$ does not isolate the items of $\{b_\ell\}_{\ell = 1}^n \cup \{d\}$ from each other, and yet there exists a strategy $\cP'_2$ of $t_2$ such that every strategy $\hat{\cP}'_1$ of $t_1$ obeys $\Pi_1(\hat{\cP}'_1, \cP'_2) \leq \Pi_1(\cP'_1, \cP'_2)$. By Lemma~\ref{lem:make_helper} we may assume that $\cP'_1$ is an $a$-helped strategy (otherwise, we can replace $\cP'_1$ with the strategy whose existence is guaranteed by this lemma).

Let $b_{\ell'}$ be an item that is not isolated by $\cP'_1$ from some other item of $\{b_\ell\}_{\ell = 1}^n \cup \{d\}$. The existence of $b_{\ell'}$ implies the existence of a redundant item $a_{\ell''}$ in $\cP'_1$ because one of the following must be true:
\begin{itemize}
	\item $b_{\ell'}$ shares a part in $\cP'_1$ with another item of $\{b_\ell\}_{\ell = 1}^n$. Since the number of $\{a_\ell\}_{\ell = 1}^n$ items is equal to the number of $\{b_\ell\}_{\ell = 1}^n$ items, there must be either a part of $\cP'_1$ that contains an item of $\{a_\ell\}_{\ell = 1}^n$ but no items of $\{b_\ell\}_{\ell = 1}^n$ or a part of $\cP'_1$ that contains two items of $\{a_\ell\}_{\ell = 1}^n$.
	\item $b_{\ell'}$ shares a part in $P'_1$ with the item $d$. Since $\cP'_1$ is $a$-helped, this part must contain also an item of $\{a_\ell\}_{\ell = 1}^n$ (which is redundant).
\end{itemize}

Consider the strategy $\tilde{\cP}'_1$ obtained from $\cP'_1$ by removing the items $a_{\ell''}$ and $b_{\ell'}$ from their original parts and placing them together in a new part. Let us analyze $R(\tilde{\cP}'_1) + R(\tilde{\cP}'_1 \cappart \cP'_2)$. Since $a_{\ell''}$ is redundant, its removal from its original part $\tilde{\cP}'_1$ does not decrease the contribution of this part to either revenue by Lemma~\ref{lem:redundent_no_loss}. Additionally, the removal of $a_{\ell''}$ leaves $b_{\ell'}$ either sharing a part in $P'_1$ with $d$ or with another item of $\{a_\ell\}_{\ell = 1}^n$ and another item of $\{b_\ell\}_{\ell = 1}^n$. In both cases, the removal of $b_{\ell'}$ does not affect the contribution of its part to $R(\tilde{\cP}'_1)$ by Observation~\ref{obs:values}. On the other hand, the removal of $b_{\ell'}$ can decrease the contribution of its part to the revenue $R(\tilde{\cP}'_1 \cappart \cP'_2)$. However, Observation~\ref{obs:values} guarantees that this decrease is at most $(3n + 1)^{-1}$. Finally, the contribution of the new part $\{a_{\ell''}, b_{\ell'}\}$ to $R(\tilde{\cP}'_1)$ is $(3n + 1)^{-1}$. Combining all these observations, we get:
\begin{align} \label{eq:change}
	R(\tilde{\cP}'_1) + R(\tilde{\cP}'_1 \cappart \cP'_2)
	\geq{} &
	R(\cP'_1) + R(\cP'_1 \cappart P_2) - (3n + 1)^{-1} + (3n + 1)^{-1} + A \\ \nonumber
	={} &
	R(\cP'_1) + R(\cP'_1 \cappart P_2) + A
	\enspace,
\end{align}
where $A$ is the contribution of parts of $\tilde{\cP}'_1 \cappart \cP'_2$ that are subsets of $\{a_{\ell''}, b_{\ell'}\}$ to $R(\tilde{\cP}'_1 \cappart \cP'_2)$. To get a contradiction it is enough to show that $A > 0$, i.e., that at least one of these parts has a positive contribution to $R(\tilde{\cP}'_1 \cappart \cP'_2)$. There are two cases to consider:
\begin{itemize}
	\item If $b_{\ell'}$ shares a part with $a_{\ell''}$ in $\cP'_2$, then the part $\{a_{\ell''}, b_{\ell'}\}$ appears in $\tilde{\cP}'_1 \cappart \cP'_2$ and contributes $(3n + 1)^{-1}$ to $R(\tilde{\cP}'_1 \cappart P_2)$.
	\item If $\cP'_2$ separates the item $b_{\ell'}$ from $a_{\ell''}$, then the part $\{b_{\ell'}\}$ appears in $\tilde{\cP}'_1 \cappart \cP'_2$ and contributes $\eps/(3n + 1)$ to $R(\tilde{\cP}'_1 \cappart P_2)$. \qedhere
\end{itemize}
\end{proof}

\begin{corollary} \label{cor:isolation}
If $\cP'_1$ and $\cP'_2$ are strategies for $t_1$ and $t_2$ that form a Nash equilibrium, then both $\cP'_1$ and $\cP'_2$ isolate the items of $\{b_\ell\}_{\ell = 1}^n \cup \{d\}$ from each other.
\end{corollary}
\begin{proof}
Since $\cP'_1$ is a part of a Nash equilibrium, it is not strictly dominated. Hence, by Lemma~\ref{lem:isolation}, it must isolate the items of $\{b_\ell\}_{\ell = 1}^n \cup \{d\}$ from each other. The corollary holds also for $\cP'_2$ by symmetry.
\end{proof}

We are now ready to analyze the price of stability of $\DSP_n$.

\begin{theorem} \label{thm:dsp_n_stability}
The price of stability of $\DSP_n$ is at least $(n + 1)/(n\eps + 1)$. Hence, for $\eps = 1/n^2$, the price of stability of $\DSP_n$ is at least $n$.
\end{theorem}
\begin{proof}
Consider an arbitrary Nash equilibrium $(\cP'_1, \cP'_2)$ of $\DSP_n$. By Corollary~\ref{cor:isolation}, both $\cP'_1$ and $\cP'_2$ must isolate the items of $\{b_\ell\}_{\ell = 1}^n \cup \{d\}$ from each other. However, every other item of $I$ must share a part with $d$ in at least one of these partitions, and thus, every item of $\{b_\ell\}_{\ell = 1}^n$ has a singleton part in $\cP'_1 \cappart \cP'_2$. Hence,
\[
	R(\cP'_1 \cappart \cP'_2)
	=
	\frac{n\eps + 1}{3n + 1}
	\enspace.
\]
Combining the last observation with Lemma~\ref{lem:opt_revenue}, we get that the price of stability of $\DSP_n$ is at least:
\[
	\frac{(n + 1)/(3n + 1)}{(n\eps + 1)/(3n + 1)}
	=
	\frac{n + 1}{n\eps + 1}
	\enspace.
	\qedhere
\]
\end{proof}

Note that Theorem~\ref{thm_unbounded_pos} follows immediately from Theorem~\ref{thm:dsp_n_stability}.

\subsection{The weaknesses of $\cS$ are inevitable} \label{ssc:inevitable}

Theorem \ref{thm_unbounded_pos} asserts that the revenue of the best equilibrium can be about $n$ times worse than the optimal revenue. 
This discouraging result raises the question of whether alternative payment rules can improve the revenue guarantees of the auctioneer. 
Unfortunately, Shapley's uniqueness theorem answers this question negatively, assuming one requires the mechanism to have some natural properties. 

We consider a family of games of the type considered in Subsection~\ref{sec_shapley}. Formally, let $\cF_m$ denote a family of $m$-player games where each player $t$ has the same finite set $A_t$ of possible strategies in all the games, one of which $\varnothing_t \in A_t$ is called the null strategy of $t$. Each game in the family is determined by an arbitrary value function $v : A_1 \times A_2 \times \ldots \times A_m \rightarrow \mathbb{R}$, and each possible such value function induces a game in $\cF_m$. Recall that a mechanism $M = (\Pi_1, \Pi_2, \ldots, \Pi_m)$ is a set of payments rules. In other words, if the players choose strategies $a_1 \in A_1, a_2 \in A_2, \ldots, a_m \in A_m$, then the payment of player $t$ under mechanism $M$ is $\Pi_t(v, a_1, a_2, \ldots, a_m)$.\footnote{Shapley's theorem is stated for cooperative games where players in the coalition can reallocate their payments within the coalition. In our non-cooperative setup, we assume side payments can be only introduced by the mechanism.}

\begin{theorem}[Uniqueness of Shapley Mechanism, cf.~\cite{Shapley53}]\label{thm_shapley_unique} 
Let $\cF_m$ be a family of games as described above. Then, the Shapley value mechanism $\cS$ is the only mechanism satisfying the following axioms:
\begin{enumerate}
\item (Normalization) For every player $t$, $\Pi_t(a) = 0$ whenever $a_t = \varnothing_t$.
\item (Anonymity) If $\cG_m \in \cF_m$ is a game with a strategy profile $a^*$ such that $v(a) = v(\varnothing)$ (recall that $\varnothing$ denotes the strategy profile $(\varnothing_1, \varnothing_2, \ldots, \varnothing_m)$) for every strategy profile $a \neq a^*$, then for every strategy profile $a$ and player $t$:
\[
	\Pi_t(a)
	=
	\begin{cases}
		0 & \text{if } a_t = \varnothing_t \enspace, \\
		\frac{v(a) - v(\varnothing)}{|\{t \in [m] \mid a_t \neq \varnothing_t \}|} & \text{otherwise} \enspace.
	\end{cases}
\]
\item (Additivity) If $\cG_m , \cH_m \in \cF_m$ are two games with value functions $v_g$ and $v_h$, then $\Pi_t^{v_g + v_h}(a) = \Pi_t^{v_g}(a) + \Pi_t^{v_h}(a)$ (where $\Pi_t^v(a)$ stands for the payment of player $t$ given strategy profile $a$ in the game defined by the value function $v$).
\end{enumerate}
\end{theorem}

We note that the second axiom (``Anonymity'') is ubiquitous in market design, and is typically enforced by market regulations prohibiting
discrimination among clients. Intuitively, this axiom says that whenever there is only one strategy profile $a^*$ which produces a value other than $v(\varnothing)$, then, when $a^*$ is played, the mechanism is required to \emph{equally} distribute the surplus $v(a^*) - v(\varnothing)$ among the participants playing a non-null 
strategy in $a^*$. Observe that violating this axiom would require private contracts with (at least some of) the players (mediators). Implementing such contracts 
defeats one of the main purposes of our mechanism, namely that it can be easily implemented in a dynamic environment having an unstable mediators population.

The above theorem was originally proved in a \emph{cooperative} setting (where players may either join a coalition or not), under slightly different
axioms. The ``anonymity'' axiom of Theorem~\ref{thm_shapley_unique} replaces the fairness and efficiency axioms of the original theorem, and is sufficient for the uniqueness proof to go through in a non-cooperative setup such as the $\DSP$ game.

%% file: discussion.tex
\section{Discussion} \label{sec:discussion}

In this paper we have considered computational and strategic aspects of auctions involving third party information mediators. Our main result for the computational point of view shows that it is NP-hard to get a reasonable approximation ratio when the three parameters of the problem are all ``large''. For the parameters $n$ and $k$ this is tight in the sense that there exists an algorithm whose approximation ratio is good when either one of these parameters is ``small''. However, we do not know whether a small value for the parameter $m$ allows for a good approximation ratio. More specifically, even understanding the approximation ratio achievable in the case $m = 2$ is an interesting open problem. Observe that the case $m = 2$ already captures (asymptotically) the largest possible price of stability and price of anarchy in the strategic setup, and thus, it is tempting to assume that this case also captures all the complexity of the computational setup.

Unfortunately, most of our results, for both the computational and strategic setups, are quite negative. The class of local experts we describe is a natural mediators class allowing us to bypass one of these negative result and get a constant approximation ratio algorithm for the computational setup. An intriguing potential avenue for future research is finding additional natural classes of mediators that allow for improved results, either under the computational or the strategic setup.